\documentclass{amsart}
\usepackage[dvips]{graphicx}
\usepackage{epsfig,enumerate,mathrsfs}
\usepackage{color}

\numberwithin{equation}{section}
\newtheorem{theorem}{Theorem}[section]

\newtheorem{lemma}[theorem]{Lemma}

\theoremstyle{definition}
\newtheorem{definition}[theorem]{Definition}
\newtheorem{example}[theorem]{Example}
\newtheorem{exercise}{Exercise}

\newcommand{\G}{\Gamma}

\newcommand{\R}{\mathbb R}
\newcommand{\Z}{\mathbb Z}

\newcommand{\Ee}{\mathcal{E}}

\def\N{\mathbb{N}}
\def\C{\mathbb{C}}

\def\KD3{\mathrm{KD}^3}

\newcommand{\Res}{{\rm Res\,\,}}

\def\beq{\begin{eqnarray}}
\def\eeq{\end{eqnarray}}
\newcommand{\nn}{\nonumber}
\newcommand{\snuu}{\sum_{n=-\infty}^\infty}
\newcommand{\sneu}{\sum_{n=1}^\infty}
\newcommand{\sumne}{\sum_{n=1}^\infty}
\newcommand{\summe}{\sum_{m=1}^\infty}
\newcommand{\sumkn}{\sum_{k=0}^\infty}
\newcommand{\sumnn}{\sum_{n=0}^\infty}
\newcommand{\sumlm}{\sum_{l=-\infty}^\infty}
\newcommand{\res}{\mbox{Res }}
\newcommand{\intl}{\int\limits}
\newcommand{\ep}{\epsilon}
\newcommand{\Cc}{\mathcal{C}}
\newcommand{\Bb}{\mathcal{B}}

\begin{document}

\title[Zeta functions and applications]{Basic zeta functions and some applications in physics}

\author{Klaus Kirsten}
\address{Department of Mathematics\\ Baylor University\\
         Waco\\ TX 76798\\ U.S.A. }
\email{Klaus$\_$Kirsten@baylor.edu}

\maketitle
\section{Introduction}
It is the aim of these lectures to introduce some basic zeta functions and their uses in the areas of the Casimir effect and Bose-Einstein condensation. A brief introduction into these areas is given in the respective sections; for recent monographs on these topics see \cite{bord01-353-1,dalv99-71-463,eliz95b,eliz94b,kirs02b,milt01b,milt04-37-209,peth02b,pita03b}. We will consider exclusively spectral zeta functions, that is zeta functions arising from the eigenvalue spectrum of suitable differential operators. Applications like those in number theory \cite{apos76b,apos90b,dave67b,titc51b} will not be considered in this contribution.

There is a set of technical tools that are at the very heart of understanding analytical properties of essentially every spectral zeta function. Those tools are introduced in Section 2 using the well-studied examples of the Hurwitz \cite{hurw82-27-86}, Epstein \cite{epst03-56-615,epst07-63-205} and Barnes zeta function \cite{barn03-19-426,barn03-19-374}. In Section 3 it is explained how these different examples can all be thought of as being generated by the same mechanism, namely they all result from eigenvalues of suitable (partial) differential operators. It is this relation with partial differential operators that provides the motivation for analyzing the zeta functions considered in these lectures. Motivations come for example from the questions "Can one hear the shape of a drum?" and "What does the Casimir effect know about a boundary?". Finally "What does a Bose gas know about its container?" The first two questions are considered in detail in Section 4. The last question is examined in Section 5 where we will see how zeta functions can be used to analyze the phenomenon of Bose-Einstein condensation. The Conclusions will point towards recent developments for the analysis of spectral zeta functions and their applications.
\section{Some basic zeta functions}
In this section we will construct analytical continuations of basic zeta functions. From these we will determine the meromorphic structure, residues at singular points and special function values.
\subsection{Hurwitz zeta function}
We start by considering a generalization of the Riemann zeta function \beq \zeta_R (s) = \sum_{n=1}^\infty \frac 1 {n^s} . \label{zetarie}\eeq
\begin{definition}
Let $s\in \C$ and $0<a<1$. Then for $\Re s >1 $ the Hurwitz zeta function is defined by \beq \zeta_H (s,a) = \sum_{n=0}^\infty \frac 1 {(n+a)^s}.\nn\eeq
\end{definition}
Clearly, we have that $\zeta_H (s,1) = \zeta _R (s)$. Results for $a=1+b>1$ follow by observing \beq \zeta_H (s,1+b) = \sumnn \frac 1 {(n+1+b)^s} = \zeta _H (s,b) - \frac 1 {b^s}.\nn\eeq
In order to determine properties of the Hurwitz zeta function, one strategy is to express it in term of 'known' zeta functions like the Riemann zeta function.
\begin{theorem}
For $0<a<1$ we have \beq \zeta_H (s,a) = \frac 1 {a^s} + \sumkn (-1)^k \frac{\Gamma (s+k)} {\Gamma (s) k!} a^k \zeta_R (s+k).\nn\eeq
\end{theorem}
\begin{proof}
Note that for $|z| <1$ we have the binomial expansion \beq (1-z)^{-s} = \sumkn \frac{\Gamma (s+k)} {\Gamma (s) k!} z^k.\nn\eeq
So for $\Re s>1$ we compute
\beq \zeta_H (s,a) &=& \frac 1 {a^s} + \sumne \frac 1 {n^s} \frac 1 {\left( 1 + \frac a n \right)^s} \nn\\
&=& \frac 1 {a^s} + \sumne \frac 1 {n^s} \sumkn (-1)^k \frac{\Gamma (s+k)} {\Gamma (s) k!}  \left( \frac a n \right)^k \nn\\
&=& \frac 1 {a^s} + \sumkn (-1)^k \frac{\Gamma (s+k)} {\Gamma (s) k!} a^k \sumne \frac 1 {n^{s+k}} \nn\\
&=& \frac 1 {a^s} + \sumkn (-1)^k \frac{\Gamma (s+k)} {\Gamma (s) k!} a^k \zeta _R (s+k),\nn\eeq
which is the assertion.
\end{proof}
From here it is seen that $s=1$ is the only pole of $\zeta_H (s,a)$ with $\res \zeta_H (1,a) =1$.

In determining certain function values of $\zeta_H (s,a)$ the following polynomials will turn out to be useful.
\begin{definition}
For $x\in \C$ we define the Bernoulli polynomials $B_n (x)$ by the equation
\beq \frac{ z e^{xz}}{e^z-1} = \sumnn \frac{B_n (x)} {n!} z^n, \quad \quad \mbox{where } |z|<2\pi.\label{bernoullipol}\eeq
\end{definition}
Examples are $B_0 (x) =1$ and $B_1 (x) = x-1/2$. The numbers $B_n (0)$ are called Bernoulli numbers and are denoted by $B_n$. Thus \beq \frac z {e^z-1} = \sum_{n=0}^\infty \frac{B_n }{n!} z^n, \quad \quad \mbox{where } |z|<2\pi.\label{bernoullinum}\eeq
\begin{lemma}\label{lem2.2}
The Bernoulli polynomials satisfy
\begin{enumerate}
\item
\beq B_n (x) = \sum_{k=0}^n { n \choose k} B_k x^{n-k} ,\nn\eeq
\item
\beq B_n (x+1) - B_n (x) = n x^{n-1} \quad \quad \mbox{if }n\geq 1 ,\nn\eeq
\item
\beq (-1)^n B_n (-x) = B_n (x) + n x^{n-1}, \nn\eeq
\item
\beq
B_n (1-x) = (-1)^n B_n (x) .\nn\eeq
\end{enumerate}
\end{lemma}
\begin{exercise} Use relations (\ref{bernoullipol}) and (\ref{bernoullinum}) to show the assertions of Lemma \ref{lem2.2}.\end{exercise}
We now establish elementary properties of $\zeta_H (s,a)$.
\begin{theorem}\label{the2.3}
For $\Re s >1$ we have
\beq \zeta_H (s,a) = \frac 1 {\Gamma (s)} \intl_0^\infty t^{s-1} \frac{e^{-at}}{1-e^{-t}} dt.\label{zetahurint}\eeq
Furthermore, for $k\in\N_0$ we have \beq \zeta_H (-k , a) = - \frac{B_{k+1}(a)} {k+1} .\nn\eeq
\end{theorem}
\begin{proof}
We use the definition of the Gamma-function and have
\beq \Gamma (s) = \intl _0^\infty u^{s-1} e^{-u} du = \lambda^s \intl_0^\infty t^{s-1} e^{-t \lambda} dt .\label{Gamma}\eeq
This shows the first part of the Theorem,
\beq \zeta_H (s,a) &=& \sumnn \frac 1 {\Gamma (s)} \intl_0^\infty t^{s-1} e^{-t (n+a)}  dt
= \frac 1 {\Gamma (s)} \intl_0^\infty t^{s-1} \sumnn e^{-t (n+a)}  dt\nn\\
&=&\frac 1 {\Gamma (s)} \intl_0^\infty t^{s-1} \frac{e^{-at}}{1-e^{-t}}dt.\nn\eeq
Furthermore we have
\beq \zeta_H (s,a) &=&
= \frac 1 {\Gamma (s)} \intl _0^\infty t^{s-2} \frac{ t e^{-ta}}{1-e^{-t}} dt \nn\\
&=& \frac 1 {\Gamma (s)} \intl_0^1 t^{s-2} \frac{ (-t) e^{-ta}} {e^{-t} -1} dt +
\frac 1 {\Gamma (s)} \intl_1^\infty t^{s-2} \frac{ (-t) e^{-ta}} {e^{-t} -1} dt .\nn\eeq
The integral in the second term is an entire function of $s$. Given the Gamma-function has singularities at $s=-k$, $k\in \N_0$, only the first term can possibly contribute to the properties $\zeta_H (-k,a)$ considered. We continue and write
\beq \frac 1 {\Gamma (s) } \intl_0^1 t^{s-2} \frac{ (-t) e^{-ta}}{e^{-t} -1} dt &=& \frac 1 {\Gamma (s)} \intl_0^1 t^{s-2} \sumnn \frac{ B_n (a)} {n!}(-t)^n dt \nn\\
&=& \frac 1 {\Gamma (s) } \sumnn \frac{B_n (a)} {n!} \frac {(-1)^n}{s+n-1} ,\nn\eeq
which provides the analytical continuation of the integral to the complex plane.
From here we observe again
\beq \res \zeta_H (1,a) &=& B_0 (a) =1 , \nn\eeq
and the second part of the Theorem
\beq
\zeta_H (-k , a) &=& \lim_{\ep\to 0} \frac 1 {\Gamma (-k+\ep )} \frac{B_{k+1} (a) } { (k+1)! } \frac{ (-1)^{k+1}} \ep \nn\\
&=& \lim_{\ep\to 0} (-1)^k k! \ep \frac{B_{k+1} (a) } { (k+1)! } \frac{ (-1)^{k+1}} \ep = - \frac{ B_{k+1} (a)} {k+1} \nn\eeq
follows.
\end{proof}
The disadvantage of the representation (\ref{zetahurint}) is that it is valid only for $\Re s >1$. This can be improved by using a complex contour integral representation. Starting point is the following representation for the Gamma-function \cite{grad65b}.
\begin{lemma}\label{lem2.4}
For $z\notin\Z$ we have
\beq \Gamma (z) = - \frac 1 {2i\sin (\pi z)} \intl _{\Cc} (-t)^{z-1} e^{-t} dt ,\nn\eeq
where the anticlockwise contour $\Cc$ consists of a circle $\Cc_3$ of radius $\epsilon < 2\pi$ and straight lines $\Cc_1$, respectively $\Cc_2$, just above, respectively just below, the $x$-axis; see Figure \ref{fig1}.
\end{lemma}
\begin{figure}[ht]
\setlength{\unitlength}{1cm}

\begin{center}
\begin{picture}(10,6.5)
\thicklines
\put(0,3){\vector(1,0){10}} \put(5.0,0){\vector(0,1){6}}
\put(8.0,5.5){{\bf $t$-plane}}
\qbezier(5.3,3.01)(8.0,3.01)(9.50,3.01) \qbezier(5.3,2.99)(8.0,2.99)(9.50,2.99)
\put(5.0,3.0){\circle{0.6}}
\put(5.0,3.0){\circle{0.59}}
\put(5.0,3.0){\circle{0.61}}
\put(5.0,3.0){\circle{0.58}}
\put(5.0,3.0){\circle{0.62}}
\put(9.30,3.2){$\Cc_1$}
\put(9.30,2.6){$\Cc_2$}
\put(4.3,3.2){$\Cc_3$}
\put(8.9,3.735){\vector(-1,0){1.7}}
\put(7.2,2.265){\vector(1,0){1.6}}
\put(7.2,3.9){$-t = e^{-i \pi} u$}
\put(7.2,1.85){$-t = e^{i \pi} u$}
\end{picture}
\caption{Contour $\Cc$ in Lemma \ref{lem2.4}.}\label{fig1}
\end{center}
\end{figure}
\begin{proof}
Assume $\Re z >1$. As the integrand remains bounded along ${\Cc}_3$, no contributions will result as $\ep \to 0$. Along ${\Cc}_1$ and ${\Cc}_2$ we parameterize as given in Figure \ref{fig1} and thus for $\Re z>1$
\beq \lim_{\epsilon\to 0}\,\,\intl_{\Cc} (-t)^{z-1} e^{-t} dt &=& \intl_\infty^0 e^{-i \pi (z-1)} u^{z-1} e^{-u} du
+ \intl_0^\infty e^{i\pi (z-1)} u^{z-1} e^{-u} du \nn\\
&=& - \intl_0^\infty u^{z-1} e^{-u} \left( e^{i\pi z} - e^{-i\pi z} \right) du\nn\\
&=& - 2i\sin (\pi z) \intl_0^\infty u^{z-1} e^{-u} du ,\nn\eeq
which implies the assertion by analytical continuation.
\end{proof}
This representation for the Gamma-function can be used to show the following result for the Hurwitz zeta function.
\begin{theorem}
For $ s \in \C $, $s\notin \N$, we have \beq \zeta_H (s,a) = - \frac{\Gamma (1-s)}{2\pi i} \intl_{\Cc} \frac{(-t)^{s-1} e^{-ta}}{1-e^{-t}} dt,\nn\eeq
with the contour $\Cc$ given in Figure \ref{fig1}.
\end{theorem}
\begin{proof}
We follow the previous calculation to note
\beq \intl_{\Cc} \frac{(-t)^{s-1} e^{-ta}}{1-e^{-t}} dt = - 2i\sin (\pi s ) \intl_0^\infty t^{s-1} \frac{e^{-ta}}{1-e^{-t}} dt,\nn\eeq
and we use \cite{grad65b} $$\sin (\pi s) \Gamma (s) = \frac \pi {\Gamma (1-s)}$$
to conclude the assertion.
\end{proof}
From here, properties previously given can be easily derived. For $s\in \Z$ the integrand does not have a branch cut and the integral can easily be evaluated using the residue theorem. The only possible singularity enclosed is at $t=0$ and to read off the residue we use the expansion
$$- (-t)^{s-2} \frac{(-t) e^{-ta}}{e^{-t}-1} = - (-t)^{s-2} \sumnn \frac{B_n (a)} {n!} (-t)^n .$$
\subsection{Barnes zeta function}\label{secbarn}
The Barnes zeta function is a multidimensional generalization of the Hurwitz zeta function.
\begin{definition}\label{defbarn}
Let $s\in \C$ with $\Re s >d$ and $c\in \R_+$, $\vec r \in \R_+^d$. The Barnes zeta function is defined as \beq \zeta_\Bb (s,c|\vec r) = \sum_{\vec m \in \N_0^d} \frac 1 {(c+\vec m \cdot \vec r)^s}.\label{eqbarn}\eeq If $c=0$ it is understood that the summation ranges over $\vec m \neq \vec 0$.
\end{definition}
For $\vec r = \vec 1_d:=(1,1,...,1,1)$, the Barnes zeta function can be expanded in terms of the Hurwitz zeta function.
\begin{example}\label{ex1}
Let us consider $d=2$ and $\vec r = (1,1)$. Then
\beq\zeta_\Bb (s,c|\vec 1_2) &=& \sum_{\vec m \in \N_0^2} \frac 1 {(c+m_1+m_2)^s}= \sumkn \frac {k+1} {(c+k)^s} = \sumkn \frac{k+c+1-c}{(c+k)^s} \nn\\
&=& \zeta_H (s-1,c) + (1-c) \zeta_H (s,c).\nn\eeq
\end{example}
\begin{example}\label{ex2}
Let $e_k^{(d)}$ be the number of possibilities to write an integer $k$ as a sum over $d$ non-negative integers. We then can write
\beq     \zeta_\Bb (s,c| \vec 1_d) = \sum_{\vec m\in \N_0^d} \frac 1 {(c+m_1+...+m_d)^s} = \sumkn e_k^{(d)} \frac 1 {(c+k)^s} .\nn\eeq
The coefficient $e_k^{(d)}$ can be determined for example as follows. Consider \beq \frac 1 {(1-x)^d} &=& \frac 1 {1-x} \cdot \cdot \cdot \frac 1 {1-x} = \left( \sum_{l_1=0}^\infty x^{l_1} \right) \cdot \cdot \cdot \left( \sum_{l_d=0}^\infty x^{l_d}\right) \nn\\
&=& \sum_{l_1=0}^\infty \cdot \cdot \cdot \sum_{l_d=0}^\infty x^{l_1+...+l_d} = \sumkn e_k^{(d)} x^k.\nn\eeq
On the other side, using the binomial expansion \beq \frac 1 {(1-x)^d} &=& \sumkn \frac{\Gamma (d+k)} {\Gamma (d) k!} x^k = \sumkn \frac{(d+k-1)!}{(d-1)! k!} x^k \nn\\
&=& \sumkn {d+k-1 \choose d-1 } x^k.\nn\eeq
This shows \beq \zeta_{\mathcal B} (s,c| \vec 1_d) = \sumkn {d+k-1 \choose d-1} \frac 1 {(c+k)^s} ,\nn\eeq
which, once the dimension $d$ is specified, allows to write the Barnes zeta function as a sum of Hurwitz zeta functions along the lines in Example \ref{ex1}.
\end{example}
It is possible to obtain similar formulas for $r_i$ rational numbers \cite{dowk94-162-633,dowk94-35-4989}.

For some properties of the Barnes zeta function the use of complex contour integral representations turns out to be the best strategy.
\begin{theorem}\label{barnrep}
We have the following representations:
\beq \zeta_\Bb (s,c| \vec r) &=& \frac 1 {\Gamma (s)} \intl_0^\infty t^{s-1} \frac{e^{-ct}}{\prod_{j=1}^d \left(1-e^{-r_j t}\right)} dt \nn\\
&=& - \frac {\Gamma (1-s)} {2\pi i} \intl_\Cc (-t)^{s-1} \frac{e^{-ct}}{\prod_{j=1}^d \left(1-e^{-r_j t}\right)} dt,\nn\eeq
with the contour $\Cc$ given in Figure \ref{fig1}.
\end{theorem}
\begin{exercise}
Use equation (\ref{Gamma}), respectively Lemma \ref{lem2.4}, to proof Theorem \ref{barnrep}.
\end{exercise}
The residues of the Barnes zeta function and its values at non-positive integers are best described using generalized Bernoulli polynomials \cite{norl22-43-121}.
\begin{definition}\label{defber}
We define the generalized Bernoulli polynomials $B_n^{(d)} (x|\vec r)$ by the equation \beq \frac{e^{-xt}}{\prod_{j=1}^d \left( 1-e^{-r_j t}\right)} = \frac{ (-1)^d}{\prod_{j=1}^d r_j} \sumnn \frac{(-t)^{n-d} }{n!} B_n^{(d)} (x|\vec r).\nn\eeq
\end{definition}
Using Definition \ref{defber}  in Theorem \ref{barnrep} one immediately obtains the following properties of the Barnes zeta function.
\begin{theorem}\label{polebarn}
We have \begin{enumerate}
\item
\beq \res \zeta_\Bb (z,c|\vec r) = \frac{ (-1)^{d+z}}{(z-1)! (d-z)! \prod_{j=1}^d r_j} B_{d-z}^{(d)} (c|\vec r ), \quad \quad z=1,2,...,d, \nn\eeq
\item
\beq \zeta_\Bb (-n , c| \vec r ) = \frac{ (-1)^d n!} {(d+n)! \prod_{j=1}^d r_j } B_{d+n} ^{(d)} (c| \vec r ) .\nn\eeq
\end{enumerate}
\end{theorem}
\begin{exercise}
Use the first representation of $\zeta_\Bb (s,c|\vec r)$ in Theorem \ref{barnrep} together with Definition \ref{defber} to show Theorem \ref{polebarn}. Follow the steps of the proof in Theorem \ref{the2.3}.
\end{exercise}
\begin{exercise}
Use the second representation of $\zeta_\Bb (s,c| \vec r)$ in Theorem \ref{barnrep} together with Definition \ref{defber} and the residue theorem to show Theorem \ref{polebarn}.
\end{exercise}
\subsection{Epstein zeta function}
We now consider zeta functions associated with sums of squares of integers \cite{epst03-56-615,epst07-63-205}.
\begin{definition}\label{epdef}
Let $s\in \C$ with $\Re s > d/2$ and $c\in \R_+$, $\vec r \in \R_+^d.$ The Epstein zeta function is defined as
\beq \zeta_\Ee (s,c| \vec r ) = \sum_{\vec m \in \Z ^d} \frac 1 {(c+r_1 m_1^2 + r_2 m_2^2+...+r_d m_d^2)^s}.\nn\eeq
If $c=0$ it is understood that the summation ranges over $\vec m \neq \vec 0$.
\end{definition}
\begin{lemma}
\label{mellinep}
For $\Re s > d/2$, we have \beq \zeta _\Ee (s,c| \vec r) = \frac 1 {\Gamma (s)} \intl_0^\infty t^{s-1} \sum_{\vec m \in \Z ^d} e^{-t (r_1 m_1^2 + ... + r_d m_d^2 + c)} dt .\nn\eeq
\end{lemma}
\begin{proof} This follows as before from property (\ref{Gamma}) of the Gamma-function.\end{proof}
As we have noted in the proof of Theorem \ref{the2.3}, it is the small-$t$ behavior of the integrand that determines residues of the zeta function and special function values. The way the integrand is written in Lemma \ref{mellinep} this $t\to 0$ behavior is not easily read off. A suitable representation is obtained by using the Poisson resummation \cite{hill62b}.
\begin{lemma}
\label{poisson}
Let $r \in \C$ with $\Re r >0$ and $t\in \R_+$, then
\beq \sum_{l=-\infty} ^\infty e^{-trl^2} = \sqrt{\frac \pi {tr} } \sumlm e^{-\frac{\pi ^2} {rt} l^2} .\nn\eeq
\end{lemma}
\begin{exercise}\label{expoi}
If $F(x)$ is continuous such that
$$\int\limits_{-\infty}^\infty \,\, | F(x) | dx < \infty,$$ then we
define its Fourier transform by
$$\hat F (u) = \int\limits_{-\infty}^\infty \,\, F(x) e^{-2\pi i x
u} \,\, dx.$$ If $$\int\limits_{-\infty}^\infty \,\, | \hat F (u)
| \,\, du < \infty , $$ then we have the Fourier inversion formula
$$F(x) = \int\limits_{-\infty} ^\infty \,\, \hat F (u) \,\,
e^{2\pi i x u} \,\, du.$$ Show the following Theorem: Let $F\in
L^1 (\R)$. Suppose that the series $$\sum_{n\in\Z} \,\,
F(n+v)$$ converges absolutely and uniformly in $v$, and that
$$\sum_{m\in\Z} \,\, | \hat F (m) | < \infty.$$ Then
$$\sum_{n\in\Z} \,\, F(n+v) = \sum_{n\in\Z} \,\, \hat F
(n) e^{2\pi i nv}.$$ Hint: Note that $$G(v) =
\sum_{n\in\Z}\,\, F(n+v)$$ is a function of $v$ of period
$1$. \end{exercise}
\begin{exercise}
Apply Exercise \ref{expoi} with a suitable function $F(x)$ to show the Poisson resummation formula Lemma \ref{poisson}.
\end{exercise}
In Lemma \ref{poisson} it is clearly seen that the only term on the right hand side that is not exponentially damped as $t\to 0$ comes from the $l=0$ term. Using the resummation formula for all $d$ sums
in Lemma \ref{mellinep}, after resumming the $\vec m = \vec 0$ term contributes
\beq \zeta_\Ee ^{\vec 0} (s,c| \vec r ) &=&
\frac 1 {\Gamma (s)} \intl_0^\infty t^{s-1} \frac{ \pi^{d/2}}{t^{d/2} \sqrt{ r_1 \cdot \cdot \cdot r_d}} e^{-ct} dt \nn\\
&=&\frac{\pi^{d/2}}{\sqrt{r_1 \cdot \cdot \cdot r_d } \,\,\Gamma (s)} \intl_0^\infty t^{s-d/2-1} e^{-ct} dt \nn\\
&=& \frac{\pi ^{d/2}}{\sqrt{r_1 \cdot \cdot \cdot r_d} } \frac{\Gamma \left( s-\frac d 2 \right)}{\Gamma (s) c^{s-d/2}} .\nn\eeq
All other contributions after resummation are exponentially damped as $t\to 0$ and can be given in terms of modified Bessel functions \cite{grad65b}.
\begin{definition}
\label{bessel}
Let $\Re z^2 >0$, then we define the modified Bessel function $K_\nu (z)$ as
\beq K_\nu (z) = \frac 1 2 \left( \frac z 2 \right)^\nu \intl_0^\infty e^{-t - \frac{z^2} {4t}}\,\, t^{-\nu -1} dt.\nn\eeq
\end{definition}
Performing the resummation in Lemma \ref{mellinep} according to Lemma \ref{poisson}, with Definition \ref{bessel} one obtains the following representation of the Epstein zeta function valid in the whole complex plane \cite{eliz94b,terr73-183-477}.
\begin{theorem}
\label{bessel1}
We have
\beq \zeta_\Ee (s,c| \vec r) &=& \frac{ \pi^{d/2}}{\sqrt{r_1 \cdot \cdot \cdot r_d}} \frac{\Gamma \left( s- \frac d 2 \right)}{\Gamma (s)} c^{\frac d 2 -s} + \frac{2 \pi ^s c^{\frac{d-2s} 4}} {\Gamma (s) \sqrt{ r_1 \cdot \cdot \cdot r_d}} \nn\\
& &\hspace{-1.0cm}\times\sum_{\vec n \in \Z ^d / \{\vec 0 \}} \left[ \frac{ n_1^2} {r_1} + ... + \frac{n_d^2} {r_d} \right] ^{\frac 1 2 \left( s- \frac d 2 \right) } K_{\frac d  2 -s} \left( 2 \pi \sqrt c \left( \frac{n_1^2} {r_1} + ... + \frac{n_d^2} {r_d} \right)^{1/2} \right) .\nn\eeq
\end{theorem}
\begin{exercise}
Show Theorem \ref{bessel1} along the lines indicated.
\end{exercise}
From Definition \ref{bessel} it is clear that the Bessel function is exponentially damped for large $\Re z^2$. As a result the above representation is numerically very effective as long as the argument of $K_{d/2-s}$ is large. The terms involving the Bessel functions are analytic for all values of $s$, the first term contains poles. As an immediate consequence of the properties of the Gamma-function one can show the following properties of the Epstein zeta function.
\begin{theorem}\label{the2.11}
For $d$ even, $\zeta_\Ee (s,c| \vec r)$ has poles at $s=\frac d 2, \frac d 2 -1, ..., 1$, whereas for $d$ odd they are located at $s= \frac d 2, \frac d 2 -1, ..., \frac 1 2, - \frac{2l+1} 2$, $l\in \N_0$. Furthermore,
\beq \res \zeta _\Ee (j,c| \vec r ) &=& \frac{ (-1)^{\frac d 2 +j} \pi ^{\frac j 2} c^{\frac d 2 -j} } {\sqrt{r_1 \cdot \cdot \cdot r_d}\,\, \Gamma (j) \Gamma \left( \frac d 2 -j +1\right)} , \nn\\
\zeta_\Ee (-p,c| \vec r ) &=& \left\{ \begin{array}{ll}
0 & \mbox{for $d$ odd}\\
\frac{ (-1)^{\frac d 2 } p! \pi ^{\frac d 2} c^{\frac d 2 +p} } {\sqrt{r_1 \cdot \cdot \cdot r_d} \,\,\Gamma \left( \frac d 2 +p +1\right)} & \mbox{for $d$ even}.\end{array} \right. \nn\eeq
\end{theorem}
\begin{exercise}
Use Theorem \ref{bessel1} and properties of the Gamma-function to show Theorem \ref{the2.11}.
\end{exercise}
This concludes the list of examples for zeta functions to be considered in what follows. A natural question is what the motivations are to consider these zeta functions. Before we describe a few aspects relating to this question let us mention how all these zeta functions, and many others, result from a common principle.
\section{Boundary value problems and associated zeta functions}
In this section we explain how the considered zeta functions, and others, are all associated with eigenvalue problems of (partial) differential operators.
\begin{example}
Let $M=[0,L]$ be some interval and consider the {\it Dirichlet} boundary value problem.
\beq P \phi_n (x) := - \frac{\partial^2}{\partial x^2} \phi _n (x) = \lambda_n \phi_n (x), \quad \quad \phi_n (0) = \phi_n (L) =0 .\nn\eeq
The solutions to the boundary value problem have the general form
\beq \phi _n (x) = A \sin (\sqrt {\lambda_n} x ) + B \cos ( \sqrt{\lambda_n} x) .\nn\eeq
Imposing the Dirichlet boundary condition shows we need
\beq \phi _n (0) = B=0, \quad \quad \phi_n (L) = A \sin (L \sqrt{\lambda_n}) = 0, \nn\eeq
which implies $$\lambda_n = \frac{n^2 \pi^2} {L^2} , \quad \quad n\in\N.$$
We only need to consider $n\in\N$ because non-positive integers lead to linearly dependent eigenfunctions. The zeta function $\zeta_P (s)$ associated with this boundary value problem is defined to be the sum over all eigenvalues raised to the power $(-s)$, namely
\beq \zeta_P (s) = \sumne \lambda_n^{-s} ,\quad \quad \Re s > \frac 1 2.\nn\eeq
So here the associated zeta function is a multiple of the zeta function of Riemann,
\beq \zeta_P (s) = \sumne \left( \frac{n \pi} L \right)^{-2s} = \left( \frac L \pi \right)^{2s} \zeta_R (2s) .\nn\eeq
\end{example}
\begin{example}
The previous example can be easily generalized to higher dimensions. We consider explicitly two dimensions; for the higher dimensional situation
see \cite{ambj83-147-1}. Let $M=\{ (x,y) | x\in [0,L_1], y\in [0,L_2]\}.$ We consider the boundary value problem with Dirichlet boundary conditions on $M$, that is
\beq P \phi_{n,m} (x,y) &=& \left(- \frac{\partial^2}{\partial x^2} - \frac{\partial ^2} {\partial y^2} + c \right) \phi _{n,m} (x,y) = \lambda_{n,m} \phi_{n,m} (x,y) , \nn\\
\phi_{n,m} (0,y) &=& \phi_{n,m} (L_1,y)= \phi_{n,m} (x,0) = \phi _{n,m} (x,L_2) =0 .\nn\eeq
Using the process of separation of variables, eigenfunctions are seen to be
\beq \phi_{n,m} (x,y) = A \sin \left(\frac{n\pi x}{L_1}\right) \sin \left(\frac{m\pi y}{L_2}\right) , \nn\eeq
with the eigenvalues $$\lambda_{n,m} = \left( \frac{ n\pi} {L_1} \right)^2 + \left( \frac{m\pi} {L_2}\right)^2 + c, \quad \quad n,m\in\N.$$
The associated zeta function therefore is
\beq \zeta_P (s) = \sumne \summe \left[ \left( \frac{n \pi } {L_1}\right)^2 + \left( \frac{ m \pi} {L_2} \right)^2 +c \right] ^{-s} , \nn\eeq
which can be expressed in terms of the Epstein zeta function given in Definition \ref{epdef} as follows,
\beq \zeta_P (s) &=& \frac 1 4 \zeta_\Ee \left( s,c\left| \left( \left( \frac \pi {L_1} \right)^2,\left( \frac \pi {L_2} \right)^2\right)\right)  \right. \nn\\
& & - \frac 1 4 \zeta_\Ee \left( s,c\left| \left( \frac \pi {L_1} \right)^2 \right)\right. -
\frac 1 4 \zeta_\Ee \left( s,c\left| \left( \frac \pi {L_2} \right)^2 \right)\right. + \frac 1 4 c^{-s}.    \eeq
\end{example}
\begin{example}
Similarly one can consider periodic boundary conditions instead of Dirichlet boundary conditions, this means the manifold $M$ is given by $M=S^1\times S^1$. In this case the eigenfunctions have to satisfy
\beq \phi_{n,m} (0,y) &=& \phi_{n,m} (L_1,y) , \quad \quad \frac \partial {\partial x} \phi_{n,m} (0,y) = \frac \partial {\partial x} \phi_{n,m} (L_1,y) , \nn\\
\phi_{n,m} (x,0) &=& \phi_{n,m} (x, L_2) , \quad \quad \frac \partial {\partial y} \phi_{n,m} (x,0) = \frac \partial {\partial y} \phi_{n,m} (x,L_2) .\nn\eeq
This shows \beq \phi_{n,m} (x,y) = A e^{i \frac{2 \pi n} {L_1} x} \,\, e^{i \frac{2\pi m} {L_2} y} , \nn\eeq
which implies for the eigenvalues $$ \lambda_{n,m} = \left( \frac{ 2\pi n} {L_1}\right)^2 + \left( \frac{ 2\pi m} {L_2} \right)^2 + c, \quad \quad (n,m) \in \Z^2.$$ The associated zeta function therefore is
\beq \zeta_P (s) = \zeta _\Ee \left(s,c| \vec r \right), \quad \quad \vec r = \left( \left(\frac{2\pi} {L_1}\right)^2,  \left(\frac{2\pi} {L_2}\right)^2\right).\nn\eeq
Clearly, in $d$ dimensions one finds \beq \zeta_P (s) = \zeta_\Ee \left( s,c| \vec r \right), \quad \quad \vec r = \left( \left( \frac{2\pi} {L_1} \right)^2 , ..., \left( \frac{2\pi } {L_d} \right)^2 \right).\nn\eeq
\end{example}
\begin{example}\label{exHO}
As a final example we consider the Schr{\"o}dinger equation of atoms in a harmonic oscillator potential. In this case $M=\R^3$, and the eigenvalue equation reads
\beq \left\{ - \frac {\hbar^2} {2m} \Delta + \frac m 2 \left( \omega_1 x^2+\omega_2 y^2 + \omega _3 z^2\right) \right\} \phi_{n_1,n_2,n_3} (x,y,z) = \lambda_{n_1,n_2,n_3} \phi_{n_1,n_2,n_3} (x,y,z).\nn\eeq
This differential equation is augmented by the condition that eigenfunctions must be square integrable, $\phi_{n_1,n_2,n_3} (x,y,z) \in \mathcal{L}^2 (\R^3).$ As is well known, this gives the eigenvalues $$\lambda_{n_1,n_2,n_3} = \hbar \omega_1 \left( n_1 + \frac 1 2 \right) + \hbar \omega_2 \left( n_2 + \frac 1 2 \right) + \hbar \omega_3 \left( n_3 + \frac 1 2 \right) , \quad (n_1,n_2,n_3)\in\N_0^3.$$
This clearly leads to the Barnes zeta function
\beq \zeta_P (s) = \zeta_\Bb (s,c| \vec r),\nn\eeq
where \beq c&=& \frac 1 2 \hbar (\omega_1 + \omega_2 + \omega_3), \quad \quad \vec r = \hbar \left(\omega_1,\omega_2,\omega_3\right).\nn\eeq
If $M=\R$ is chosen the Hurwitz zeta function results.
\end{example}
 The above examples illustrate how the zeta functions considered in Section 2 are all related in a natural way to eigenvalues of specific boundary value problems. In fact, zeta functions in a much more general context are studied in great detail. For our purposes the relevant setting is the setting of Laplace-type operators on a Riemannian manifold $M$, possibly with a boundary $\partial M$. Laplace-type means the operator $P$ can be written as $$P = - g^{jk} \nabla _j^V \nabla_k^V - E,$$ where $g^{jk}$ is the metric of $M$, $\nabla^V$ is the connection on $M$ acting on a smooth vector bundle $V$ over $M$, and where $E$ is an endomorphism of $V$. Imposing suitable boundary conditions, eigenvalues $\lambda_n$ and eigenfunctions $\phi_n$ do exist, $$P \phi_n (x) = \lambda_n \phi_n (x), $$ and assuming $\lambda_n>0$ the zeta function is defined to be $$\zeta _P (s) = \sum_{n=1}^\infty \lambda_n ^{-s}$$ for $\Re s$ sufficiently large. If there are modes with $\lambda_n=0$ those have to be excluded from the sum. Also, if finitely many eigenvalues are negative the zeta function can be defined by choosing nonstandard definitions of the principal value for the argument of complex numbers, but we will not need to consider those cases.
\section{(Some) Motivations to consider zeta functions}
There are many situations where properties of zeta functions in the above context of Laplace-type operators are needed. In the following we present a few of them, but many more can be found for example in the context of number theory \cite{apos76b,apos90b,dave67b,titc51b} and quantum field theory \cite{bord01-353-1,buch92b,byts03b,byts96-266-1,dett92-377-252,dunn08-41-304006,dunn05-94-072001,eliz95b,espo94b,espo98b,kirs02b,sach92-65-652}.
\subsection{Can one hear the shape of a drum?}
Let $M$ be a two-dimensional membrane representing a drum with boundary $\partial M$. The drum is fixed along its boundary. Then possible vibrations of the drum and its fundamental tones are described by the eigenvalue problem \beq \left( - \frac{\partial^2}{\partial x^2} - \frac {\partial^2} {\partial y^2} \right) \phi_n (x,y) = \lambda_n \phi_n (x,y) , \quad \quad \left.\phi_n (x,y) \right|_{(x,y) \in \partial M} =0.\nn\eeq
Here, $(x,y)$ denotes the variables in the plane, the eigenfunctions $\phi_n (x,y)$ describe the amplitude of the
vibrations and $\lambda_n$ its fundamental tones. In 1966 Kac \cite{kac66-73-1} asked if just
by listening with a perfect ear, so by knowing all the fundamental tones
$\lambda_n$, it is possible to hear the shape of the drum. One problem in answering this question is, of course, that in
general it will be impossible to write down the eigenvalues $\lambda_n$
in a closed form and to read off relations with the shape of the drum
directly. Instead one has to organize the spectrum intelligently in
form of a spectral function to reveal relationships between the eigenvalues
and the shape of the drum. In this context a particularly fruitful spectral
function is the heat kernel
$$
K(t) = \sum_{n=1}^\infty e^{-\lambda_n t},
$$
which as $t$ tends to zero clearly diverges. Given that some relations between the fundamental tones and properties of the drum are hidden in the $t\to 0$ behavior let us consider this asymptotic behavior very closely. Before we come back to the setting of the drum, let us use a few examples to get an idea what the structure of the $t\to 0$ behavior of the heat kernel is expected to be.
\begin{example}\label{heatcircle} Let $M=S^1$ be the circle with circumference $L$ and let $P = - \partial^2/\partial x^2$. Imposing periodic boundary conditions eigenvalues are $$\lambda_k = \left( \frac{2\pi k} L \right)^2 , \quad \quad k\in\Z ,$$ and the heat kernel reads
\beq K_{S^1} (t) = \sum_{k=-\infty}^\infty e^{- \left( \frac{2\pi k} L \right)^2t} .\nn\eeq
From Lemma \ref{poisson} we find the $t\to 0$ behavior $$K_{S^1} (t) = \frac 1 {\sqrt{4\pi t}} L +
\left( \mbox{exponentially damped terms}\right). $$ Note that with obvious notation this could be written as
$$K_{S^1} (t) =\frac 1 {\sqrt{4\pi t}} \mbox{vol} (M) + \left( \mbox{exponentially damped terms}\right). $$
\end{example}
\begin{example} \label{heattorus} The heat kernel for the $d$-dimensional manifold $M=S^1\times \cdot \cdot \cdot \times S^1$ with $P=-\Delta$ clearly gives a product of the above and thus $$K_M (t) = K_{S^1} (t) \times \cdot \cdot \cdot \times K_{S^1} (t) = \frac 1 {(4\pi t)^{d/2}} \mbox{vol} (M) +
\left( \mbox{exponentially damped terms}\right). $$
\end{example}
\begin{example} \label{massivetorus} To avoid the impression that there is always just one term that is not exponentially damped consider $M$ as above but $P=-\Delta + m^2$. Then
\beq K(t) &=& e^{-m^2 t} K_M (t) = e^{-m^2 t} \left( \frac 1 {(4\pi t)^{d/2}} \mbox{vol} (M) +
\mbox{exponentially damped terms} \right)\nn\\
&=& \frac 1 {(4\pi)^{d/2}} \mbox{vol} (M) \sum_{\ell =0}^\infty \frac{(-1)^\ell} {\ell !} m^{2\ell } t^{\ell - \frac d 2}
+ \left(\mbox{exponentially damped terms}\right) .\nn\eeq
\end{example}
In fact, the structure of the heat kernel observed in this last example is the structure observed for the general class of Laplace-type operators.
\begin{theorem}\label{heatnobound}
Let $M$ be a $d$-dimensional smooth compact Riemannian manifold without boundary and let $$P= - g^{jk} \nabla_j^V \nabla_k^V - E ,$$ where $g^{jk}$ is the metric of $M$, $\nabla^V$ is the connection on $M$ acting on a smooth vector bundle $V$ over $M$, and where $E$ is an endomorphism of $V$.
Then as $t\to 0$, $$K (t) \sim \sum_{k=0}^\infty a_k \,\,t^{k-d/2}$$ with the so-called heat kernel coefficients $a_k$.
\end{theorem}
\begin{proof} See, e.g., \cite{gilk95b}. \end{proof}
In Example \ref{massivetorus} one sees that $$a_k = \frac 1 {(4\pi)^{d/2}} \frac{(-1)^k}{k!} m^{2k} \mbox{vol}(M).$$ In general, the heat kernel coefficients are significantly more complicated and they depend upon the geometry of the manifold $M$ and the endomorphism $E$ \cite{gilk95b}.

Up to this point we have only considered manifolds without boundary. In order to consider in more detail questions relating to the drum, let us now see what relevant changes in the structure of the small-$t$ heat kernel expansion occur if boundaries are present.
\begin{example}\label{heatinterval}
Let $M=[0,L]$ and $P=-\partial^2/\partial x^2$ with Dirichlet boundary conditions imposed. Normalized eigenfunctions are then given by $$\varphi_\ell (x) = \sqrt{\frac 2 L} \sin\left( \frac{\pi \ell x} L \right)$$ and the associated eigenvalues are $$\lambda_\ell = \left( \frac{\pi \ell} L \right)^2, \quad \quad \ell \in \N.$$ Using Lemma \ref{poisson} this time we obtain \beq K(t) = \frac 1 {\sqrt{4\pi t}} \mbox{vol}(M) - \frac 1 2 + (\mbox{exponentially damped terms}).\label{intbounddir}\eeq
Notice that in contrast to previous results we have integer and half-integer powers in $t$ occurring.
\end{example}
\begin{exercise}
There is a more general version of the Poisson resummation formula than the one given in Lemma \ref{poisson}, namely
\beq \sum_{\ell =-\infty} ^\infty e^{-t (\ell +c)^2} = \sqrt{\frac \pi t } \sum_{\ell =-\infty}^\infty e^{-\frac{\pi^2} t \ell^2 - 2\pi i \ell c} .\label{genpoisson}\eeq Apply Exercise \ref{expoi} with a suitable function $F(x)$ to show equation (\ref{genpoisson}).
\end{exercise}
\begin{exercise}\label{localheat}
Consider the setting described in Example \ref{heatinterval}. The local heat kernel is defined as the solution of the equation
\beq \left( \frac \partial {\partial t} - \frac{\partial^2}{\partial x^2}\right) K(t,x,y) =0 \nn\eeq
with the initial condition
\beq \lim_{t\to 0} K(t,x,y) = \delta (x,y) .\nn\eeq
In terms of the quantities introduced in Example \ref{heatinterval} it can be written as \beq K(t,x,y) = \sum_{\ell =1}^\infty \varphi _\ell (x) \varphi_\ell (y) e^{-\lambda_\ell t} .\nn\eeq
Use the resummation (\ref{genpoisson}) for $K(t,x,y)$ and the fact that $$K(t) = \intl _0^L K(t,x,x) dx$$ to rediscover the above result (\ref{intbounddir}).
\end{exercise}
\begin{exercise} \label{heatboundint} Let $M=[0,L]$ and $$ P= - \frac{\partial ^2} {\partial x^2} + m^2$$ with Dirichlet boundary conditions imposed. Find the small-$t$ asymptotics of the heat kernel.\end{exercise}
\begin{exercise} \label{highdbound} Let $M=[0,L] \times S^1\times\cdot \cdot \cdot \times S^1$ be a $d$-dimensional manifold and $$P=-\frac{\partial^2}{\partial x^2} + m^2.$$ Impose Dirichlet boundary conditions on $[0,L]$ and periodic boundary conditions on the circle factors. Find the small-$t$ asymptotics of the heat kernel.\end{exercise}
As the above examples and exercises suggest, one has the following result.
\begin{theorem}\label{heatbound} Let $M$ be a $d$-dimensional smooth compact Riemannian manifold with smooth boundary and
let $$P= - g^{jk} \nabla_j^V \nabla_k^V - E ,$$ where $g^{jk}$ is the metric of $M$, $\nabla^V$ is the connection on $M$ acting on a smooth vector bundle $V$ over $M$, and where $E$ is an endomorphism of $V$. We impose Dirichlet boundary conditions.
Then as $t\to 0$, $$K (t) \sim \sum_{k=0,1/2,1,...}^\infty a_k \,\, t^{k-d/2}$$ with the heat kernel coefficients $a_k$.
\end{theorem}
\begin{proof} See, e.g., \cite{gilk95b}. \end{proof}
As for the manifold without boundary case, Theorem \ref{heatnobound}, the heat kernel coefficients depend upon the geometry of the manifold $M$ and the endomorphism $E$, and in addition on the geometry of the boundary. Note, however, that in contrast to Theorem \ref{heatnobound} the small-$t$ expansion contains integer and half-integer powers in $t$.

The same structure of the small-$t$ asymptotics is found for other boundary conditions like Neumann or Robin, see \cite{gilk95b}, and the coefficients then also depend on the boundary condition chosen.
In particular, for Dirichlet boundary conditions one can show the identities
\beq a_0 = (4 \pi)^{-d/2} \mbox{vol}(M), \quad \quad a_{1/2} = (4\pi )^{-(d-1)/2} \left( - \frac 1 4 \right) \mbox{vol}(\partial M),\label{leadcoef}\eeq
a result going back to McKean and Singer \cite{mcke67-1-43}. In the context of the drum, what the formula shows is that
by listening with a perfect ear one can indeed hear certain properties
like the area of the drum and the circumference of its boundary. But as has been
shown by Gordon, Webb and Wolpert \cite{gord92-27-134}, one cannot hear all details of the
shape.
\begin{exercise}
Use Exercise \ref{highdbound} to verify the general formulas (\ref{leadcoef}) for the heat kernel coefficients.
\end{exercise}
Instead of using the heat kernel coefficients to make the above statements, one could equally well have used zeta function properties for equivalent statements. Consider the setting of Theorem \ref{heatbound}.
The associated zeta function is given by $$ \zeta_P (s) = \sum_{n=1}^\infty \lambda_n^{-s},$$ where it follows from Weyl's law \cite{weyl12-71-441,weyl15-39-1} that this series is convergent for $\Re s>d/2$. The zeta function is related with the heat kernel by \beq \zeta_P (s) = \frac 1 {\Gamma (s)} \intl _0^\infty  t^{s-1} K(t) dt , \label{zetaheat}\eeq
where equation (\ref{Gamma}) has been used. This equation allows us to relate residues and function values at certain points with the small-$t$ behavior of the heat kernel. In detail,\beq \mbox{Res } \zeta_P (z) &=& \frac{a_{\frac d 2 -z}} {\Gamma (z)} , \quad \quad z=\frac d 2 , \frac{d-1} 2 , ..., \frac 1 2, - \frac{2n+1} 2 , n \in\N_0, \label{zeheres}\\
\zeta_P (-q) &=& (-1)^q q! a_{\frac d 2 +q} , \quad \quad q\in\N_0. \label{zehefun}\eeq
Keeping in mind the vanishing of the heat kernel coefficients $a_k$ with half-integer index for $\partial M=\emptyset$, see Theorem \ref{heatnobound}, this means for $d$ even the poles are actually located only at $z=d/2,d/2-1,...,1$. In addition, for $d$ odd we get $\zeta_P (-q)=0$ for $q\in\N_0$.
\begin{exercise} Use Theorem \ref{heatbound} and proceed along the lines indicated in the proof of Theorem \ref{the2.3} to show equations (\ref{zeheres}) and (\ref{zehefun}). \end{exercise}

Going back to the setting of the drum properties of the zeta function relate with the geometry of the surface. In particular, from (\ref{leadcoef}) and (\ref{zeheres}) one can show the identities
\beq \res \zeta_P (1) = \frac{\mbox{vol}(M)} {4\pi}, \quad \res \zeta _P \left( \frac 1 2 \right) = - \frac{\mbox{vol}(\partial M)} {2\pi}, \nn\eeq
and the remarks below equation (\ref{leadcoef}) could be repeated.
\subsection{What does the Casimir effect know about a boundary?}
We next consider an application in the context of quantum field theory in finite systems. The importance of this topic lies in the fact that
in recent years, progress in many fields has been triggered by the
continuing miniaturization of all kinds of technical devices. As
the separation between components of various systems tends towards
the nanometer range, there is a growing need to understand every
possible detail of quantum effects due to the small sizes involved.

Very generally speaking, effects resulting from the finite
extension of systems and from their precise form are known as the
Casimir effect. In modern technical devices this effect is
responsible for up to 10$\%$ of the forces encountered in
microelectromechanical systems
\cite{chan01-87-211801,chan01-291-1941}. Casimir forces are of
direct practical relevance in nanotechnology where, e.g., sticking
of mobile components in micromachines might be caused by them
\cite{serr95-4-193}. Instead of fighting the occurrence of the
effect in technological devices, the tendency is now to try and
take technological advantage of the effect.

Experimental progress in recent years has been impressive and for
some configurations allows for a detailed comparison with
theoretical predictions. The best tested situations are those of
parallel plates \cite{bres02-88-041804} and of a plate and a
sphere
\cite{chan01-291-1941,chen06-74-022103,lamo97-78-5,lamo00-84-5673,mohi98-81-4549};
recently also a plate and a cylinder has been considered
\cite{brow05-72-052102,emig06-96-080403}. Experimental data and theoretical
predictions are in excellent agreement, see, e.g.,
\cite{bord01-353-1,decc05-318-37,klim06-39-6485,lamo05-68-201}.
This interplay between theory and experiments, and the intriguing
technological applications possible, are the main reasons for the
heightened interest in this effect in recent years.

In its original form, the effect refers to the situation of two
uncharged, parallel, perfectly conducting plates. As predicted by
Casimir \cite{casi48-51-793}, the plates should attract with a
force per unit area, $F(a)\sim 1/a^4$, where $a$ is the distance
between the plates. Two decades later Boyer \cite{boye68-174-1764} found
a repulsive pressure of magnitude $F (R) \sim 1/R^4$ for a
perfectly conducting spherical shell of radius $R$. Up to this day an intuitive understanding of the opposite signs found is lacking.
One of the main questions in the context of the Casimir effect therefore is how the occurring forces depend on the geometrical properties of the system considered. Said differently, the question is "What does the Casimir effect know about a boundary?" In the absence of general answers one approach consists in accumulating further knowledge by adding bits of understanding based on specific calculations for specific configurations. Several examples will be provided in this section and we will see the dominant role the zeta functions introduced play. However, before we come to specific settings let us briefly introduce the zeta function regularization of the Casimir energy and force that we will use later on.

We will consider the Casimir effect in a quantum field theory of a non-interacting scalar field under {\it external} conditions. The action in this case is \cite{itzy80b}
\beq S[\Phi ] = - \frac 1 2 \intl_M  \Phi (x) \left( \Delta - V(x) \right) \Phi (x) \,\, dx\label{lec29a}\eeq
describing a scalar field $\Phi (x)$ in the background potential $V(x)$. We assume the Riemannian manifold $M$ to be of the form $M=S^1 \times M_s$, where the circle $S^1$ of radius $\beta$ is used to describe finite temperature $T=1/\beta$ and $M_s$, in general, is a $d$-dimensional Riemannian manifold with boundary. For the action (\ref{lec29a}) the corresponding field equations are \beq (\Delta -V(x)) \Phi (x) =0 . \label{lec29b}\eeq
If $M_s$ has a boundary $\partial M_s$, these equations of motion have to be supplemented by boundary conditions on $\partial M_s$. Along the circle, for a scalar field, periodic boundary conditions are imposed.

Physical properties like the Casimir energy of the system are conveniently  described by means of the path-integral functionals \beq Z [V] = \int  e^{-S [\Phi ] } \,\, D\Phi , \label{lec29c}\eeq
where we have neglected an infinite normalization constant, and the functional integral is to be taken over all fields satisfying the boundary conditions. Formally, equation (\ref{lec29c}) is easily evaluated to be
\beq  \Gamma [V] = - \ln Z[V] = \frac 1 2 \ln \det \left[ ( - \Delta + V(x) )/\mu^2\right] , \label{lec29d}\eeq
where $\mu$ is an arbitrary parameter with dimension of a mass to adjust the dimension of the arguments of the logarithm.
\begin{exercise}
In order to motivate equation (\ref{lec29d}) show that for $P$ a positive definite Hermitian $(N\times N)$-matrix one has
\beq \intl_{\R^n} e^{-\frac 1 2 (x,Px)} (dx) = (\det P)^{-1/2},\nn\eeq where $(dx) = d^nx (2\pi)^{-n/2}$. For $P=-\Delta + V(x)$ and interpreting the scalar product $(x,Px)$ as an $L^2 (M)$-product, one is led to (\ref{lec29d}) by identifying $D\Phi$ with $(dx)$.
\end{exercise} Clearly equation (\ref{lec29d}) is purely formal because the eigenvalues $\lambda_n$ of $-\Delta + V(x)$ grow without bound for $n\to\infty$ and thus expression (\ref{lec29d}) needs further explanations.

In order to motivate the basic definition let $P$ be a Hermitian $(N\times N)$-matrix with positive eigenvalues $\lambda_n$. Clearly $$\left.\ln \det P = \sum_{n=1}^N \ln \lambda_n = - \frac d {ds} \sum_{n=1}^N \lambda_n^{-s} \right|_{s=0} =  \left. -\frac d {ds} \zeta_P (s) \right|_{s=0} , $$ and the determinant of $P$ can be expressed in terms of the zeta function associated with $P$. This very same definition, namely \beq \ln \det P = - \zeta_P ' (0) \label{lec29e}\eeq with \beq \zeta_P (s) = \sum_{n=1}^\infty \lambda_n^{-s}\label{lec29ea}\eeq is now applied to differential operators as in (\ref{lec29d}). Here, the series representation is valid for $\Re s$ large enough, and in (\ref{lec29e}) the unique analytical continuation of the series to a neighborhood about $s=0$ is used.

 This definition was first used by the mathematicians Ray and Singer \cite{ray71-7-145} to give a definition of the Reidemeister-Franz torsion. In physics, this regularization scheme took its origin in ambiguities of dimensional regularization when applied to quantum field theory in curved spacetime \cite{dowk76-13-3224,hawk77-55-133}. For applications beyond the ones presented here see, e.g., \cite{buch92b,byts03b,dett92-377-252,dunn08-41-304006,dunn05-94-072001,espo94b,espo98b,sach92-65-652}.

The quantity $\Gamma [V]$ is called the effective action and the argument $V$ indicates the dependence of the effective action on the external fields. The Casimir energy is obtained from the effective action via \beq E=\frac \partial {\partial \beta} \Gamma [V] = - \frac 1 2 \frac \partial {\partial \beta} \zeta_{P/\mu^2} ' (0) . \label{lec29f}\eeq
Here, we will only consider the zero temperature Casimir energy \beq E_{Cas} = \lim_{\beta \to \infty} E \label{lec29fa}\eeq and we will next derive a suitable representation for $E_{Cas}$. We want to concentrate on the influence of boundary conditions and therefore we set $V(x) =0$. The relevant operator to be considered therefore is $$P = - \frac{\partial^2}{\partial \tau^2} - \Delta_s$$ where $\tau \in S^1$ is the imaginary time and $\Delta_s$ is the Laplace operator on $M_s$. In order to analyze the zeta function associated with $P$ we note that eigenfunctions, respectively eigenvalues, are of the form
\beq \phi_{n,j} (\tau , y) &=& \frac 1 \beta e^{\frac{2\pi in} \beta \tau } \varphi_j (y) , \nn\\
\lambda_{n,j} &=& \left( \frac{ 2\pi n} \beta \right)^2 + E_j ^2, \quad \quad n\in\Z,\nn\eeq
with $$-\Delta_s \varphi _j (y) = E_j ^2 \varphi_j (y),$$
where $y\in M_s$.
For the non-selfinteracting case considered here, $E_j$ are the one-particle energy eigenvalues of the system. The relevant zeta function therefore has the structure \beq \zeta_P (s) = \sum_{n=-\infty}^\infty \sum_{j=1}^\infty \left( \left( \frac{2\pi n} \beta \right)^2 + E_j^2\right)^{-s} . \label{lec29g}\eeq
We repeat the analysis outlined previously, namely we use equation (\ref{Gamma}) and we apply Lemma \ref{poisson} to the $n$-summation. In this process the zeta function $$\zeta_{P_s} (s) = \sum_{j=1}^\infty E_j^{-2s}$$ and the heat kernel $$K_{P_s} (t) = \sum_{j=1}^\infty e^{-E_j^2 t} \sim \sum_{k=0,1/2,1,...}^\infty a_k \,\, t^{k-\frac d 2}$$ of the spatial section are the most natural quantities to represent the answer,
\beq
\zeta_P (s) &=& \frac 1 {\G (s)} \snuu \intl_0^\infty t^{s-1}
     e^{-\left(\frac{2\pi n } \beta \right)^2 t } K_{P_s} (t)\,dt \nn\\
&=& \frac \beta {\sqrt{4\pi}} \frac{\G (s-1/2)}{\G (s)} \zeta_{P_s} (s-1/2)
            \nn\\
  & &
      + \frac \beta {\sqrt{\pi}\,\,\G (s)} \sneu \intl_0^\infty
 t^{s-3/2} e^{-\frac{n^2\beta^2}{4t}} K_{P_s} (t)\,dt .\nn
\eeq
For the Casimir energy we need ($D=d+1$)
\beq
\zeta ' _{P / \mu^2} (0) &=& \zeta_P ' (0) +\zeta_P (0) \ln \mu^2 \nn\\
&= &-\beta \left( FP\,\,\zeta_{P_s} (-1/2) +
        2(1-\ln 2) \Res \zeta_{P_s} (-1/2) \right.\nn\\
  & &\left. -\frac 1 \beta \zeta_P (0)
        \ln \mu^2 \right)
   +   \frac \beta {\sqrt{\pi} } \sneu \intl_0^\infty t^{-3/2}
                 e^{-\left(\frac{n^2\beta^2}{4t}\right)}K_{P_s} (t) \,dt\nn\\
  &=& -\beta \left(FP\,\,\zeta_{P_s} (-1/2)
-   \frac 1 {\sqrt{4\pi}} a_{D/2} \left[(\ln\mu^2) +
      2 (1-\ln 2 ) \right]\right) \nn\\
   & &        +
               \frac \beta {\sqrt{\pi} } \sneu \intl_0^\infty t^{-3/2}
                 e^{-\left(\frac{n^2\beta^2}{4t}\right)}K_{P_s} (t)\, dt,
               \label{lec29h}
\eeq
with the finite part $FP$ of the zeta function and where equations (\ref{zeheres}) and (\ref{zehefun}) together with the fact that $$K_M (t) = K_{S^1} (t) \,\, K_{P_s} (t)$$ have been used, in particular
\beq \mbox{Res }\zeta_{P_s} \left( - \frac 1 2 \right) = - \frac{ a_{D/2}} {2\sqrt \pi} , \quad \quad \zeta_P (0) = \frac \beta {\sqrt{4\pi}} a_{D/2} . \label{relzetcoe}\eeq
At $T=0$ we obtain for the Casimir energy, see equations (\ref{lec29f}) and (\ref{lec29fa}),
\beq
E_{Cas} = \lim_{\beta \to\infty} E = \frac 1 2 FP\,\,\zeta_{P_s} (-1/2)
      -\frac 1 {2\sqrt{4\pi}} a_{D/2}
        \ln \tilde \mu^2 , \label{lec29i}
\eeq
with the scale $\tilde \mu = (\mu e / 2)$. Equation (\ref{lec29i}) implies that as long as $a_{D/2} \neq 0$ the Casimir energy contains a finite ambiguity and renormalization issues need to be discussed. Note from (\ref{relzetcoe}) that whenever $\zeta_{P_s} (-1/2)$ is finite no ambiguity exists because $a_{D/2}=0$. In the specific examples chosen later we will make sure that these ambiguities are absent and therefore a discussion of renormalization will be unnecessary.

In a purely formal calculation one essentially is also led to equation (\ref{lec29i}). As mentioned, in the quantum field theory of a free scalar field the eigenvalues of a Laplacian are the square of the energies of the quantum fluctuations.
Writing the Casimir energy as (one-half) the sum over the energy of all quantum fluctuations one has
\beq E_{Cas} = \frac 1 2 \sum_{k=0}^\infty \lambda_k^{1/2}, \label{lec30}\eeq
and a formal identification 'shows'
\beq E_{Cas} = \frac 1 2 \zeta_{P_s} \left( - \frac 1 2 \right).\label{lec31}\eeq
Clearly, the expression (\ref{lec30}) is purely formal as the series diverges. However, when $\zeta_{P_s} (-1/2)$ turns out to be finite this formal identification yields the correct result. Otherwise, the ambiguities given in (\ref{lec29i}) remain as discussed above.

An alternative discussion leading to definition (\ref{lec29i}) can be found in \cite{blau88-310-163}.

As a first example let us consider the configuration of two parallel plates a distance $a$ apart analyzed originally by Casimir \cite{casi48-51-793}. For simplicity we concentrate on a scalar field instead of the electromagnetic field and we impose Dirichlet boundary conditions on the plates. The boundary value problem to be solved therefore is $$-\Delta u_k (x,y,z) = \lambda_k u_k (x,y,z) $$ with $$u_k(0,y,z) = u_k(a,y,z) =0.$$ For the time being, we compactify the $(y,z)$-directions to a torus with perimeter length $R$ and impose periodic boundary conditions in these directions. Later on, the limit $R\to \infty$ is performed to recover the parallel plate configuration. Using separation of variables one obtains normalized eigenfunctions in the form
$$u_{\ell_1 \ell_2 \ell} (x,y,z) = \sqrt{\frac 2 {aR^2}} \sin \left( \frac{\pi \ell} a x\right) e^{i \frac {2\pi\ell_1 y} R }  e^{i \frac {2\pi\ell_2 z} R } $$ with eigenvalues $$\lambda_{\ell _1 \ell _2 \ell } = \left( \frac{ 2\pi \ell _1} R \right)^2+\left( \frac{ 2\pi \ell _2} R \right)^2+
\left( \frac{ \pi \ell } a \right)^2 , \quad (\ell _1, \ell _2) \in \Z^2, \quad \ell \in \N.$$
This means we have to study the zeta function
\beq \zeta (s) = \sum_{(\ell _1, \ell _2) \in \Z^2} \,\,\sum_{\ell =1}^\infty \left[ \left( \frac{2\pi \ell_1} R \right)^2 + \left( \frac{2\pi \ell_2} R \right)^2 + \left( \frac{\pi \ell} a \right)^2 \right]^{-s} . \label{lec32}\eeq
As $R\to\infty$ the Riemann sum turns into an integral and we compute using polar coordinates in the $(y,z)$-plane \beq \zeta (s) &=& \left( \frac R {2\pi}\right)^2 \sum_{\ell =1}^\infty \,\,\int\limits_{-\infty} ^\infty \,\, \int\limits_{-\infty} ^\infty  \left[ k_1^2+k_2^2 + \left(\frac{\pi \ell} a \right)^2 \right]^{-s} \,\,dk_2 \,\, dk_1\nn\\
&=& \left( \frac R {2\pi}\right)^2 \sum_{\ell =1}^\infty 2\pi \int\limits_0^\infty  k \left[ k^2 + \left( \frac{\pi \ell} a \right)^2 \right]^{-s} \,\,dk \nn\\
&=& \left.\frac{R^2} {2\pi} \frac 1 {2 (1-s)} \sum_{\ell =1}^\infty \left[ k^2 + \left( \frac{\pi \ell} a \right)^2 \right]^{-s+1} \right|_0^\infty \nn\\
&=& - \frac{R^2} {4\pi (1-s)} \sum_{\ell =1}^\infty \left( \frac{\pi \ell} a \right)^{2 (-s+1)} \nn\\
&=& - \frac{R^2}{4\pi (1-s)} \left( \frac \pi a \right)^{2-2s} \zeta_R (2s-2).\nn\eeq
Setting $s=-1/2$ as needed for the Casimir energy we obtain \beq \zeta \left( - \frac 1 2 \right) = - \frac{R^2} {4\pi } \,\, \frac 2 3 \,\, \left( \frac \pi a \right)^3 \zeta_R (-3) = - \frac{ R^2 \pi^2} { 720 a^3} . \label{lec33} \eeq
The resulting Casimir force {\it per area} is \beq F_{Cas} = -\frac \partial {\partial a} \frac{E_{Cas}}{R^2} = -\frac{\pi^2} {480 a^4} . \label{lec34}\eeq
Note, that this computation takes into account only those quantum fluctuations from between the plates. But in order to find the force acting on the, say, right plate the contribution from the right to this plate also has to be counted. To find this part we place another plate at the position $x=L$ where at the end we take $L\to\infty$. Following the above calculation, we simply have to replace $a$ by $L-a$ to see that the associated zeta function produces $$ \zeta \left( - \frac 1 2 \right) = - \frac{R^2 \pi^2}{720 (L-a)^3} $$ and the contribution to the force on the plate at $x=a$ reads $$F_{Cas} =  \frac{\pi^2} {480 (L-a)^4}.$$
This shows the plate at $x=a$ is always attracted to the closer plate. As $L\to\infty$ it is seen that equation (\ref{lec34}) also describes the total force on the plate at $x=a$ for the parallel plate configuration.
\begin{exercise}
Consider the Casimir energy that results in the previous discussion when the compactification length $R$ is kept finite. Use
Lemma \ref{bessel1} to give closed answers for the energy and the resulting force. Can the force change sign depending on $a$ and $R$?
\end{exercise}
More realistically plates will have a finite extension. An interesting setting that we are able to analyze with the tools provided are pistons. These have received an increasing amount of interest because they allow the unambiguous prediction of forces \cite{cava04-69-065015,hert05-95-250402,kirs09-79-065019,mara07-75-085019,teo09-672-190}.

Instead of having parallel plates let us consider a box with side lengths $L_1,L_2$ and $L_3$. Although it is possible to find the Casimir force acting on the plate at $x=L_1$ resulting from the interior of the box, the exterior problem has remained unsolved until today. No analytical procedure is known that allows to obtain the Casimir energy or force for the outside of the box. This problem is avoided by adding on another box with side lengths $L-L_1,L_2$ and $L_3$ such that the wall at $x=L_1$ subdivides the bigger box into two chambers. The wall at $x=L_1$ is assumed to be movable and is called the piston. Each chamber can be dealt with separately and total energies and forces are obtained by adding up the two contributions. Assuming again Dirichlet boundary conditions and starting with the left chamber, the relevant spectrum reads
 \beq \lambda_{\ell _1 \ell _2 \ell _3} = \left( \frac{\pi \ell _1} {L_1}\right)^2+\left( \frac{\pi \ell _2} {L_2}\right)^2+\left( \frac{\pi \ell _3} {L_3}\right)^2, \quad \quad \ell_1, \ell _2,\ell _3 \in \N, \label{lec34a}\eeq and the associated zeta function is \beq \zeta (s) = \sum_{\ell _1,\ell _2,\ell _3 \in \N} \left[\left( \frac{\pi \ell _1} {L_1}\right)^2+\left( \frac{\pi \ell _2} {L_2}\right)^2+\left( \frac{\pi \ell _3} {L_3}\right)^2\right]^{-s}. \label{lec35}\eeq
One way to proceed is to rewrite (\ref{lec35}) in terms of the Epstein zeta function in Definition \ref{epdef}.
\begin{exercise}
Use Lemma \ref{bessel1} in order to find the Casimir energy for the inside of the box with side lengths $L_1,L_2$ and $L_3$ and with Dirichlet boundary conditions imposed.
\end{exercise}
Instead of using Lemma \ref{bessel1} we proceed as follows. We write first \beq \zeta (s) &=& \frac 1 2 \sum_{\ell _1 = - \infty}^\infty \,\,\sum_{\ell _2, \ell_3 =1}^\infty \left[ \left( \frac{\pi \ell _1} {L_1}\right)^2 + \left( \frac{ \pi \ell _2} {L_2} \right)^2 + \left( \frac{\pi \ell_3} {L_3} \right)^2 \right] ^{-s} \nn\\
& -&  \frac 1 2 \sum_{\ell _2, \ell _3 =1}^\infty \left[ \left( \frac{\pi \ell_2} {L_2} \right)^2 + \left( \frac{\pi \ell _3}{L_3}\right)^2 \right]^{-s} . \label{lec36} \eeq
This shows that it is convenient to introduce \beq \zeta _{\mathcal C}(s) = \sum_{\ell _2, \ell _3 =1}^\infty \left[ \left( \frac{\pi \ell_2} {L_2}\right)^2 + \left( \frac{\pi \ell _3}{L_3}\right)^2 \right]^{-s} . \label{lec37}\eeq
We note that this could be expressed in terms of the Epstein zeta function given in Definition \ref{epdef}. However, it will turn out that this is unnecessary.

Also, to simplify the notation let us introduce $$\mu_{\ell_2 \ell_3} ^2 = \left( \frac{ \pi \ell_2}{L_2}\right)^2 + \left( \frac{ \pi \ell_3}{L_3}\right)^2 . $$ Using equation (\ref{Gamma}) for the first line in (\ref{lec36}) we continue \beq \zeta (s) &=& \frac 1 {2 \Gamma (s)} \sum_{\ell_1=-\infty}^\infty \,\,\sum_{\ell_2, \ell_3 =1}^\infty \int\limits_0^\infty  t^{s-1} \exp \left\{ - t \left[ \left( \frac{\pi \ell_1}{L_1}\right)^2 + \mu_{\ell_2\ell_3}^2 \right] \right\} dt \nn\\
&- &\frac 1 2 \zeta_{\mathcal C} (s).\nn\eeq
We now apply the Poisson resummation in Lemma \ref{poisson} to the $\ell_1$-summation and therefore we get
\beq \zeta (s) &=& \frac {L_1} {2\sqrt \pi \,\,\Gamma (s)} \sum_{\ell_1 =-\infty} ^\infty \,\,\sum_{\ell_2, \ell_3=1}^\infty \int\limits_0^\infty  t^{s-\frac 3 2 } \exp \left\{ - \frac{L_1^2 \ell_1^2} t - t \mu_{\ell_2\ell_3}^2\right\} dt\nn\\
& -& \frac 1 2 \zeta_{\mathcal C} (s).\label{lec37a}\eeq
The $\ell_1 =0$ term gives a $\zeta _{\mathcal C}$-term, the $\ell_1\neq 0$ terms are rewritten using (\ref{bessel}). The outcome reads \beq \zeta (s) &=& \frac{L_1 \Gamma \left( s-\frac 1 2 \right)}{2 \sqrt \pi \,\,\Gamma (s)} \zeta _{\mathcal C} \left( s- \frac 1 2 \right) - \frac 1 2 \zeta _{\mathcal C} (s) \label{lec38} \\
&+& \frac{2 L_1 ^{s+ \frac 1 2 }}{\sqrt \pi \,\,\Gamma (s)} \sum_{\ell_1, \ell _2, \ell _3 =1}^\infty \left(\frac{\ell_1^2}{\mu_{\ell_2 \ell_3}^2} \right)^{\frac 1 2 \left( s- \frac 1 2\right)} K_{\frac 1 2 -s } \left( 2 L_1 \ell_1 \mu_{\ell_2 \ell_3} \right).\nn\eeq We need the zeta function about $s=-1/2$ in order to find the Casimir energy and Casimir force.

Let $s=-1/2+\epsilon$. In order to expand equation (\ref{lec38}) about $\epsilon =0$ we need to know the pole structure of $\zeta_{\mathcal C} (s)$. From equation (\ref{bessel1}) it is expected that $\zeta_{\mathcal C} (s)$ has at most a first order pole at $s=-1/2$ and that it is analytic about $s=-1$. So for now let us simply assume the structure \beq \zeta_{\mathcal C} \left( - \frac 1 2 + \epsilon \right) &=& \frac 1 \epsilon \mbox{Res } \zeta_{\mathcal C} \left(  - \frac 1 2 \right) + \mbox{FP } \zeta_{\mathcal C} \left( - \frac 1 2 \right) + {\mathcal O} (\epsilon ), \nn\\
\zeta _{\mathcal C} (-1+\epsilon ) &=& \zeta_{\mathcal C} (-1) + \epsilon \zeta_{\mathcal C} ' (-1) + {\mathcal O} (\epsilon ^2), \nn\eeq
where $\mbox{Res } \zeta_{\mathcal C} (-1/2)$ and $\mbox{FP }\zeta_{\mathcal C} (-1/2)$ will be determined later.
With this structure assumed, we find \beq \zeta \left( - \frac 1 2 + \epsilon \right) &=& \frac 1 \epsilon \left( \frac{L_1}{4\pi} \zeta_{\mathcal C} (-1) - \frac 1 2 \mbox{Res } \zeta_{\mathcal C} \left( - \frac 1 2 \right) \right)\nn\\
& & + \frac{L_1} {4\pi} \left( \zeta_{\mathcal C} ' (-1)
+ \zeta _{\mathcal C} (-1) ( \ln 4 -1)\right) - \frac 1 2 \mbox{FP } \zeta_{\mathcal C} \left( - \frac 1 2 \right) \label{lec39} \\
& &-\frac 1 \pi \sum_{\ell_1, \ell_2, \ell_3 =1}^\infty \left| \frac{\mu_{\ell_2 \ell_3}}{\ell_1} \right| K_1 \left( 2 L_1 \ell_1 \mu_{\ell_2 \ell_3}\right) .\nn\eeq
This shows that the Casimir energy for this setting is unambiguously defined only if
$\zeta_{\mathcal C} (-1) =0$ and $\mbox{Res }\zeta_{\mathcal C} (-1/2)=0$.
\begin{exercise} \label{exep2} Show the following analytical continuation for $\zeta_{\mathcal C} (s)$:
\beq \zeta_{\mathcal C} (s) &=& - \frac 1 2 \left( \frac{L_3} \pi \right)^{2s} \zeta_R (2s) + \frac{L_2 \Gamma \left( s - \frac 1 2 \right)}
{2 \sqrt \pi \,\,\Gamma (s)} \left( \frac{L_3} \pi \right)^{2s-1} \zeta_R (2s-1) \label{lec40} \\
& &+ \frac{2 L_2 ^{s+1/2}}{\sqrt \pi \,\,\Gamma (s)} \sum_{\ell_2=1}^\infty \sum_{\ell _3 =1}^\infty \left( \frac{\ell_2 L_3}{\pi \ell_3}\right)^{s-1/2} K_{\frac 1 2 -s} \left( \frac{2\pi L_2 \ell_2 \ell_3} {L_3} \right) \nn.\eeq
Read off that $\zeta_{\mathcal C} (-1) = \mbox{Res } \zeta _{\mathcal C} (-1/2) =0$. \end{exercise}
Using the results from Exercise \ref{exep2} the Casimir energy, from equation (\ref{lec39}), can be expressed as
\beq E_{Cas} &=& \frac{L_1}{8\pi} \zeta_{\mathcal C} ' (-1) - \frac 1 4 \mbox{FP } \zeta_{\mathcal C} \left( - \frac 1 2 \right) \label{lec41}\\
&-& \frac 1 {2\pi} \sum_{\ell_1, \ell _2, \ell_3 =1}^\infty \left| \frac{\mu_{\ell_2 \ell_3}}{\ell_1} \right| K_1 (2L_1 \ell_1 \mu_{\ell _2 \ell_3}) . \nn\eeq
\begin{exercise} \label{exenergy} Use representation (\ref{lec40}) to give an explicit representation of the Casimir energy (\ref{lec41}).\end{exercise}
For the force this shows \beq F_{Cas} = - \frac 1 {8\pi} \zeta_{\mathcal C} ' (-1) + \frac 1 {2\pi } \sum_{\ell_1 , \ell_2, \ell _3 =1}^\infty \left| \frac{ \mu_{\ell_2 \ell_3}}{\ell_1} \right| \frac \partial {\partial L_1} K_1 (2 L_1 \ell_1 \mu_{\ell_2 \ell_3} ) .\label{lec42}\eeq
\begin{exercise} \label{exkelvin} Use Definition \ref{bessel} to show that $K_\nu (x)$ is a monotonically decreasing function for $x\in\R_+$.\end{exercise}
\begin{exercise} Determine the sign of $\zeta_{\mathcal C} ' (-1)$. What is the sign of the Casimir force as $L_1 \to \infty$? What about $L_1 \to 0$?\end{exercise}
Remember that the results given describe the contributions from the interior of the box only. The contributions from the right chamber are obtained by replacing $L_1$ with $L-L_1$. This shows for the right chamber
\beq E_{Cas} &=& \frac{L-L_1}{8\pi} \zeta _{\mathcal C} ' (-1) - \frac 1 4 \mbox{FP } \zeta _{\mathcal C} \left( - \frac 1 2 \right) \nn\\
& -& \frac 1 {2\pi} \sum_{\ell_1, \ell_2, \ell_3=1}^\infty \left| \frac{ \mu_{\ell _2 \ell_3}}{\ell_1} \right| K_1 (2 (L-L_1) \ell_1 \mu_{\ell_2 \ell_3} ) ,\nn\\
F_{Cas} &=& \frac 1 {8\pi} \zeta_{\mathcal C} ' (-1) + \frac 1 {2\pi} \sum_{\ell_1 , \ell_2, \ell_3 =1}^\infty \left| \frac{\mu_{\ell_2 \ell_3}}{\ell_1} \right| \frac{\partial}{\partial L_1} K_1 (2 (L-L_1) \ell_1 \mu_{\ell_2 \ell_3}) .\nn\eeq
 Adding up, the total force on the piston is
 \beq F_{Cas}^{tot} &=& \frac 1 {2\pi } \sum_{\ell_1 , \ell_2, \ell _3 =1}^\infty \left| \frac{ \mu_{\ell_2 \ell_3}}{\ell_1} \right|
 \frac \partial {\partial L_1}K_1 (2 L_1 \ell_1 \mu_{\ell_2 \ell_3} )\nn\\&+& \frac 1 {2\pi} \sum_{\ell_1 , \ell_2, \ell_3 =1}^\infty \left| \frac{\mu_{\ell_2 \ell_3}}{\ell_1} \right| \frac{\partial}{\partial L_1} K_1 (2 (L-L_1) \ell_1 \mu_{\ell_2 \ell_3}) . \label{lec43}\eeq
 This shows, using the results of Exercise \ref{exkelvin}, that the piston is always attracted to the closer wall.

 Although we have presented the analysis for a piston with rectangular cross-section, our result in fact holds in much greater generality. The fact that we analyzed a rectangular cross-section manifests itself in the spectrum (\ref{lec34a}), namely the part $$\left( \frac{\pi \ell_2} {L_2} \right)^3 + \left( \frac{\pi \ell_3}{L_3} \right)^2 $$ is a direct consequence of it. If instead we had considered an arbitrary cross-section ${\mathcal C}$, the relevant spectrum had the form $$\lambda_{\ell _1 i} = \left( \frac{\pi \ell_1} {L_1} \right)^2 + \mu_i^2,$$ where, assuming still Dirichlet boundary conditions on the boundary of the cross-section ${\mathcal C}$, $\mu_i^2$ is determined from $$ \left.\left( - \frac{ \partial ^2}{\partial y^2} - \frac{\partial ^2}{\partial z^2} \right) \phi_i (y,z) = \mu_i^2 \phi_i (y,z) , \quad \quad \phi_i (y,z) \right|_{(y,z) \in \partial {\mathcal C}} =0.$$ Proceeding in the same way as before, replacing $\mu_{\ell_2 \ell _3}$ with $\mu_i$ and introducing $\zeta_{\mathcal C} (s)$ as the zeta function for the cross-section, $$\zeta_{\mathcal C} (s) = \sum_{i=1}^\infty \mu_i^{-2s},$$ equation (\ref{lec37a}) remains valid, as well as equations (\ref{lec38}) and (\ref{lec39}). So also for an arbitrary cross-section the total force on the piston is described by equation (\ref{lec43}) with the replacements given and the piston is attracted to the closest wall.
 \begin{exercise} In going from equation (\ref{lec37a}) to (\ref{lec38}) the fact that $\mu_{\ell_2 \ell_3 }^2 > 0$ is used. Above we used $\mu_i^2>0$
 which is true because we imposed Dirichlet boundary conditions.
 Modify the calculation if boundary conditions are chosen (like Neumann boundary conditions) that allow for $d_0$ zero modes $\mu_i^2 =0$ \cite{kirs09-79-065019}.\end{exercise}
We have presented the piston set-up for three spatial dimensions, but a similar analysis can be performed in the presence of extra dimensions \cite{kirs09-79-065019}. Once this kind of calculation is fully understood for the electromagnetic field it is hoped that future high-precision measurements of Casimir forces for simple configurations such as parallel plates can serve as a window into properties of the dimensions of the universe that are somewhat hidden from direct observations.

 As we have seen for the example of the piston, there are cases where an unambiguous prediction of Casimir forces is possible. Of course the set-up we have chosen was relatively simple and for many other configurations even the sign of Casimir forces is unknown. This is a very active field of research; some references are \cite{bord01-353-1,emig07-99-170403,full09-42-155402,milt01b,milt04-37-209,scha09-102-060402}. Further discussion is provided in the Conclusions.
 \section{Bose-Einstein condensation of Bose gases in traps}
 We now turn to applications in statistical mechanics. We have chosen to apply the techniques in a quantum mechanical system described by the Schr{\"o}dinger equation \beq\left( - \frac{\hbar^2}{2m} \Delta + V(x,y,z)\right) \phi_k (x,y,z) = \lambda_k \phi_k (x,y,z), \label{lec50a}\eeq
 that is we consider a gas of quantum particles of mass $m$ under the influence of the potential $V(x,y,z)$. Specifically, later we will consider in detail the harmonic oscillator potential $$V(x,y,z) = \frac m 2 \left( \omega_1 x^2 + \omega_2 y^2 + \omega _3 z^2\right) $$ briefly mentioned in Example \ref{exHO}, as well as a gas confined in a finite cavity.

 Thermodynamical properties of a {\it bose} gas, which is what we shall consider in the following, are described by the (grand canonical) partition sum \beq q = - \sum_{k=0}^\infty \ln \left( 1-e^{-\beta (\lambda_k -\mu) } \right) , \label{lec50}\eeq where $\beta$ is the inverse temperature and $\mu$ is the chemical potential. We assume the index $k=0$ labels the unique ground state, that is, the state with smallest energy eigenvalue $\lambda_0$. From this partition sum all thermodynamical properties are obtained. For example the particle number is \beq N= \left.\frac 1 \beta \frac{\partial q}{\partial \mu} \right|_{T,V} = \sum_{k=0}^\infty \frac 1 {e^{\beta (\lambda_k - \mu)} -1} , \label{lec51}\eeq
 where the notation $(\partial q/\partial \mu|_{T,V})$ indicates that the derivative has to be taken with temperature $T$ and volume $V$ kept fixed.
 The particle number is the most important quantity for the phenomenon of Bose-Einstein condensation.
 Although this phenomenon was predicted more than 80 years ago \cite{bose24-26-178,eins24-22-261}
 it was only relatively recently experimentally verified \cite{ande95-269-198,brad95-75-1687,davi95-75-3969}.
Bose-Einstein condensation is one of the most interesting properties of a system of bosons. Namely, under certain conditions it is possible to have a phase transition at a critical value of the temperature in which all of the bosons can condense into the ground state. In order to understand at which temperature the phenomenon occurs a detailed study of $N$, or alternatively $q$, is warranted. This is the subject of this section.

We first note that from the fact that the particle number in each state has to be non-negative it is clear that $\mu<\lambda_0$
 has to be imposed.
 It is seen in (\ref{lec50}) that as $\beta \to 0$ (high temperature limit) the behavior of $q$ cannot be easily understood. But contour integral techniques together with the zeta function information provided makes the analysis feasible and it will allow for the determination of the critical temperature of the bose gas.

Let us start by noting that from $$\ln (1-x) = - \sum_{n=1}^\infty \frac{x^n} n , \quad \quad \mbox{for }|x|<1,$$ the partition sum can be rewritten as \beq q = \sum_{n=1}^\infty \sum_{k=0}^\infty \frac 1 n e^{-\beta (\lambda_k - \mu ) n } . \label{lec52}\eeq
 The $\beta \to 0$ behavior is best found using the following representation of the exponential.
\begin{exercise} \label{expcon} Given that $$\lim_{|y| \to \infty} | \Gamma (x+iy)|
\,\,e^{\frac \pi 2 |y|}\,\, |y|^{\frac 1 2 -x} = \sqrt{2\pi} ,
\quad x,y\in\R ,$$ and $$\Gamma (z) = \sqrt{ 2\pi} e^{\left(
z-\frac 1 2 \right) \log z \,\, - z } \left( 1+ o (1)\right),$$ as
$|z| \to \infty$, show that
\beq e^{-a} = \frac 1 {2\pi i} \int\limits_{\sigma-i\infty} ^{\sigma
+ i\infty} \,\, a^{-t}\,\, \Gamma (t) \,\,dt,\label{lec53}\eeq valid for $\sigma >
0$, $|\mbox{arg}\,\, a| < \frac \pi 2 - \delta $, $0<\delta \leq
 \pi /2$. \end{exercise}
Before we apply this result to the partition sum (\ref{lec52}) let us use a simple example to show how this formula allows us to determine asymptotic behavior of certain series in a relatively straightforward fashion. From Lemma \ref{poisson} we know that
\beq \sum_{\ell =1}^\infty e^{-\beta \ell^2} &=& \frac 1 2 \sum_{\ell =-\infty}^\infty e^{-\beta \ell^2} - \frac 1 2 \nn\\
&=& \frac 1 2 \sqrt{\frac \pi \beta }- \frac 1 2 +  \sqrt{\frac \pi \beta} \sum_{\ell =1}^\infty e^{-\frac{\pi^2} \beta \ell^2} .\label{lec52c}\eeq
As $\beta \to 0$ it is clear that the series on the left diverges and Lemma \ref{poisson} shows that the leading behavior is described by a $1/\sqrt{\beta}$ term, followed by a constant term, followed by exponentially damped corrections. Let us see how we can easily find the polynomial behavior as $\beta \to 0$ from (\ref{lec53}). We first write \beq \sum_{\ell =1}^\infty e^{-\beta \ell ^2} = \sum_{\ell =1}^\infty \frac 1 {2\pi i} \intl _{\sigma - i\infty} ^{\sigma + i\infty} (\beta \ell ^2)^{-t} \Gamma (t)dt. \nn\eeq
Here, $\sigma > 0$ is assumed by Exercise \ref{expcon}. However, in order to be allowed to interchange summation and integration we need to impose $\sigma >1/2$ and find
\beq \sum_{\ell =1}^\infty e^{-\beta \ell ^2} &=& \frac 1 {2\pi i} \intl_{\sigma-i\infty} ^{\sigma+i\infty} \beta^{-t} \Gamma (t) \sum_{\ell =1}^\infty \ell ^{-2t}dt \nn\\
&=& \frac 1 {2\pi i} \intl_{\sigma - i\infty} ^{\sigma + i \infty}  \beta^{-t} \Gamma (t) \zeta_R (2t) dt.\nn\eeq
In order to find the small-$\beta$ behavior, the strategy now is to shift the contour to the left. In doing so we cross over poles of the integrand generating polynomial contributions in $\beta$. For this example, the right most pole is at $t=1/2$ (pole of the zeta function of Riemann) and the next pole is at $t=0$ (from the Gamma-function). Those are all singularities present as $\zeta_R (-2n)=0$ for $n\in\N$.
Therefore, \beq \sum_{\ell =1}^\infty e^{-\beta \ell ^2} &=& \beta^{-\frac 1 2} \,\,\Gamma \left( \frac 1 2 \right) \frac 1 2 + \beta^0 \cdot 1 \cdot \zeta_R (0) +\frac 1 {2\pi i} \intl_{\tilde \sigma -i\infty}^{\tilde \sigma + i\infty}  \beta^{-t} \Gamma (t) \zeta_R (2t)dt \nn\\
&=& \frac 1 2 \sqrt{\frac \pi \beta} - \frac 1 2 + \frac 1 {2\pi i}\intl_{\tilde \sigma -i\infty}^{\tilde \sigma + i\infty}  \beta^{-t} \Gamma (t) \zeta_R (2t) dt,\nn\eeq where $\tilde \sigma <0$ and where contributions from the horizontal lines between $\tilde \sigma \pm i\infty$ and $\sigma\pm i\infty$ are neglected. For the remaining contour integral plus the neglected horizontal lines one can actually show that they will produce the exponentially damped terms as given in (\ref{lec52c}). How exactly this actually happens has been described in detail in \cite{eliz89-40-436}.
\begin{exercise}
Argue how $\sum_{n=1}^\infty \,\, e^{-\beta n^\alpha},$
$\beta
>0, \alpha >0$, behaves as $\beta \to 0$ by using the above procedure. Determine the leading three terms in the expansion assuming that the contributions from the contour at infinity can be neglected.
\end{exercise}
\begin{exercise}
Find the leading three terms of the small-$\beta$ behavior of
$$\sum_{n=1}^\infty \log \left( 1-e^{-\beta n}\right)$$
assuming that the contributions from the contour at infinity can be neglected.
\end{exercise}
Let us next apply the above ideas to the partition sum (\ref{lec52}). As a further warmup, for simplicity, let us first set $\mu =0$. Not specifying $\lambda_k$ for now and using $$\zeta (s) = \sum_{k=0}^\infty \lambda_k^{-s}$$
for $\Re s >M$ large enough to make this series convergent,
we write
\beq q &=& \sum_{n=1}^\infty \sum_{k=0}^\infty \frac 1 n e^{-\beta \lambda_k n} \nn\\
&=& \sum_{n=1}^\infty \sum_{k=0}^\infty \frac 1 n \frac 1 {2\pi i} \intl_{\sigma - i\infty}^{\sigma+i\infty}  (\beta \lambda_k n)^{-t} \Gamma (t) dt \nn\\
&=& \frac 1 {2\pi i} \intl_{\sigma - i\infty}^{\sigma + i \infty} \beta^{-t} \Gamma (t) \left( \sum_{n=1}^\infty n^{-t-1}\right) \left( \sum_{k=0}^\infty \lambda_k ^{-t} \right) dt\nn\\
&=& \frac 1 {2\pi i} \intl _{\sigma -i\infty}^{\sigma + i\infty} \beta^{-t} \Gamma (t) \zeta_R (t+1) \zeta (t)dt .\nn\eeq
Here $\sigma >M$ is needed for the algebraic manipulations to be allowed.
It is clearly seen that the integrand has a double pole at $t=0$. The right most pole (at $M$) therefore comes from $\zeta (t)$, and the location of this pole determines the leading $\beta \to 0$ behavior of the partition sum.

For the harmonic oscillator potential, in the notation of Example \ref{exHO}, the Barnes zeta function occurs and we have
\beq q = \frac 1 {2\pi i} \intl_{\sigma-i\infty}^{\sigma + i\infty}\beta^{-t} \Gamma (t) \zeta_R (t+1) \zeta_{\mathcal B} (t,c|\vec r)dt .\label{lec60}\eeq
The location of the poles and its residues are known for the Barnes zeta function, see Definition \ref{defber} and Theorem \ref{polebarn}, in particular one has $$\mbox{Res } \zeta_{\mathcal B} (3,c| \vec r) = \frac 1 {2 \hbar^3\Omega^3}, $$ where, as is common, the geometric mean of the oscillator frequencies $$\Omega = (\omega_1 \omega_2 \omega_3)^{1/3}$$ has been used. The leading order of the partition sum therefore is
\beq q = \frac{\pi^4}{90} \frac 1 {(\beta \hbar \Omega )^3} + {\mathcal O} (\beta^{-2}).\nn\eeq
\begin{exercise}
Use Definition \ref{defber} and Theorem \ref{polebarn} to find the subleading order of the small-$\beta$ expansion of the partition sum $q$.
\end{exercise}
\begin{exercise}
Consider the harmonic oscillator potential in $d$ dimensions and find the leading and subleading order of the small-$\beta$ expansion of the partition sum $q$.
\end{exercise}
If instead of considering a bose gas in a trap we consider the gas in a finite three-dimensional cavity $M$ with boundary $\partial M$ we have to augment the Schr{\"o}dinger equation (\ref{lec50a}) by boundary conditions. We choose Dirichlet boundary conditions and thus the results for the heat kernel coefficients (\ref{leadcoef}) are valid.

From equation (\ref{zeheres}) we conclude furthermore that the rightmost pole of $\zeta (s)$ is located at $s=3/2$ and that $$\res \zeta \left( \frac 3 2 \right) = \frac{a_0}{\Gamma \left( \frac 3 2 \right) } = \frac{\mbox{vol}(M)}{4\pi^2},$$ furthermore the next pole is located at $s=1$. For this case, the leading order of the partition sum therefore is $$q=\frac 1 {(4\pi\beta)^{3/2}} \zeta_R \left( \frac 5 2 \right) \mbox{vol}(M) + {\mathcal O} (\beta^{-1}).$$ One way to read this result is that the bose gas does know the volume of its container because it can be found from the partition sum. This is completely analogous to the statement for the drum where we used the heat kernel instead of the partition sum.

Subleading orders of the partition sum reveal more information about the cavity, see the following exercise. But as for the drums, the gas does not know all the details of the shape of the cavity because there are different cavities leading to the same eigenvalue spectrum \cite{gord92-27-134}. Those cavities cannot be distinguished by the above analysis.
\begin{exercise}
Consider a bose gas in a $d$-dimensional cavity $M$ with boundary $\partial M$. Use (\ref{leadcoef}) and (\ref{zeheres}) to find the leading and subleading order of the small-$\beta$ expansion of the partition sum $q$. What does the bose gas know about its container, meaning what information about the container can be read of from the high-temperature behavior of the partition sum?
\end{exercise}
In order to examine the phenomenon of Bose-Einstein condensation we have to consider non-vanishing chemical potential. Close to the phase transition, as we will see,
more and more particles have to reside in the ground state and the value of the chemical potential will be close to the smallest eigenvalue, which is the 'critical' value for the chemical potential,
$\mu_c = \lambda_0$. Near the phase transition, for the expansion to be established, it will turn out advantageous to rewrite $\lambda_k-\mu$ such that the small quantity $\mu_c-\mu$ appears, $$\lambda_k-\mu = \lambda_k - \mu_c + \mu_c - \mu = \lambda_k - \lambda_0 + \mu_c - \mu.$$
Given the special role of the ground state, we separate off its contribution and write
\beq q = q_0 +\sum_{n=1}^\infty {\sum_{k=1}^{\,\,\,\infty }}  \,\,\frac 1 n \,e^{-\beta n (\lambda_k -\lambda_0 )} \,\, e^{-\beta n (\mu_c -\mu)} .\nn\eeq
Note that the $k$-sum starts with $k=1$, which means that the ground state is not included in this summation. Employing the representation (\ref{lec53}) only to the first exponential factor and proceeding as before we obtain \beq q = q_0 + \frac 1 {2\pi i} \intl_{\sigma - i\infty} ^{\sigma + i \infty}  \beta^{-t} \Gamma (t) \mbox{Li}_{1+t} \left( e^{-\beta (\mu_c - \mu)}\right) \zeta_0 (t) dt, \label{lec61}\eeq
with the polylogarithm \beq \mbox{Li}_n (x) = \sum_{\ell =1}^\infty \frac{x^\ell} {\ell^n}, \label{lec62}\eeq
and the spectral zeta function \beq \zeta_0 (s) = {\sum_{k=1}^{\,\,\,\infty}}(\lambda_k - \lambda_0)^{-s}.\nn\eeq
In order to determine the small-$\beta$ behavior of expression (\ref{lec61}) let us discuss the pole structure of the integrand.
Given $\mu_c - \mu >0$, the polylogarithm $\mbox{Li}_{1+t} (e^{-\beta (\mu_c -\mu )})$ does not generate any poles. Concentrating on the harmonic oscillator, we find
\beq \mbox{Res } \zeta_0 (3) &=& \frac 1 {2(\hbar \Omega)^3}, \nn\\
\mbox{Res } \zeta_0 (2) &=& \frac 1 {2\hbar^2} \left( \frac 1 {\omega_1 \omega_2} + \frac 1 {\omega_1 \omega_3} + \frac 1 {\omega_2 \omega_3} \right).\nn\eeq
Note that $\zeta_0 (s)$ is the Barnes zeta function as given in Definition \ref{defbarn} with $c=0$ where we have to exclude $\vec m= \vec 0$ from the summation. However, clearly the residues at $s=3$ and $s=2$ can still be obtained from Theorem \ref{polebarn} with $c\to 0$ taken.

Shifting the contour to the left we now find
\beq q &=& q_0 + \frac 1 {(\beta \hbar \Omega)^3} \mbox{Li} _4 \left( e^{-\beta (\mu_c -\mu)}\right)\nn\\
& &+ \frac 1 {2(\beta \hbar)^2} \mbox{Li}_3 \left( e^{-\beta (\mu_c -\mu)}\right) \left( \frac 1 {\omega_1 \omega_2} + \frac 1 {\omega_1 \omega_3} + \frac 1 {\omega_2 \omega_3} \right) + ...\nn\eeq
In order to find the particle number $N$ we need the relation for the polylogarithm \beq \frac{\partial \mbox{Li}_n (x)} {\partial x} = \frac 1 x \mbox{Li}_{n-1} (x),\nn\eeq which follows from (\ref{lec62}). So
\beq N &=& N_0 +\frac 1 {(\beta \hbar \Omega)^3} \mbox{Li} _3 \left( e^{-\beta (\mu_c -\mu)}\right)\nn\\
& &+ \frac 1 {2(\beta \hbar)^2} \mbox{Li}_2 \left( e^{-\beta (\mu_c -\mu)}\right) \left( \frac 1 {\omega_1 \omega_2} + \frac 1 {\omega_1 \omega_3} + \frac 1 {\omega_2 \omega_3} \right) + ...\nn\eeq
\begin{exercise} \label{propol} Use (\ref{lec53}) and (\ref{lec62}) to show $$ \mbox{Li}_n \left( e^{-x}\right) = \zeta_R (n) - x \zeta_R (n-1) + ...$$ valid for $n>2$. What does the subleading term look like for $n=2$?
\end{exercise}
As the critical temperature is approached $\mu\to\mu_c$ and with Exercise \ref{propol} the particle number close to the transition temperature becomes \beq N &=& N_0 +\frac {\zeta_R (3)} {(\beta \hbar \Omega)^3}
 + \frac {\zeta_R (2)} {2(\beta \hbar)^2} \left( \frac 1 {\omega_1 \omega_2} + \frac 1 {\omega_1 \omega_3} + \frac 1 {\omega_2 \omega_3} \right) + ...\label{lec63}\eeq
The second and third term give the number of particles in the {\it excited} levels (at high temperature close to the phase transition).

The critical temperature is defined as the temperature where all excited levels are completely filled such that lowering the temperature the ground state population will start to build up. This means the defining equation for the critical temperature $T_c = 1/\beta _c$ in the approximation considered is
\beq N = \frac 1 {(\beta_c \hbar \Omega)^3} \zeta_R (3) + \frac 1 {2(\beta_c\hbar)^2} \zeta_R (2) \left( \frac 1 {\omega_1 \omega_2} + \frac 1 {\omega_1 \omega_3} + \frac 1 {\omega_2 \omega_3} \right).\eeq
Solving for $\beta_c$ one finds \beq T_c = T_0 \left\{ 1-\frac{\zeta_R (2)} {3 \zeta_R (3)^{2/3}} \,\,\delta \,\, N^{-1/3}\right\}.\nn\eeq
Here, $T_0$ is the critical temperature in the bulk limit ($N\to\infty$)
\beq T_0 = \hbar \Omega \left( \frac N {\zeta_R (3)} \right)^{1/3}\nn\eeq
and \beq \delta = \frac 1 2 \Omega^{2/3} \left( \frac 1 {\omega_1 \omega_2} + \frac 1 {\omega_1 \omega_3} + \frac 1 {\omega_2 \omega_3} \right).\nn\eeq
Different approaches can be used to obtain the same answers \cite{gros95-50-921,gros95-208-188,haug97-225-18,haug97-55-2922}.

If only a few thousand particles are used in the experiment the finite-$N$ correction is actually quite important. For example the first successful experiments on Bose-Einstein condensates were done with rubidium \cite{ande95-269-198} at frequencies $\omega_1 = \omega_2 = 240 \pi / \sqrt 8$ s$^{-1}$ and $\omega_3 = 240 \pi$s$^{-1}$. With $N=2000$ one finds $T_c \sim 31.9$nK$=0.93\,\, T_0$ \cite{kirs96-54-4188}, a significant correction compared to the thermodynamical limit.
\begin{exercise}
Consider the bose gas in a $d$-dimensional cavity. Find the particle number and the critical temperature along the lines described for the harmonic oscillator. What is the correction to the critical temperature caused by the finite size of the cavity? (For a solution to this problem see \cite{kirs99-59-158}.)
\end{exercise}
\section{Conclusions}
In these lectures some basic zeta functions are introduced and used to analyze the Casimir effect and Bose-Einstein condensation for particular situations. The basic zeta functions considered are the Hurwitz, the Barnes and the Epstein zeta function. Although these zeta functions differ from each other they have one property in common: they are based upon a sequence of numbers that is explicitly known and given in closed form. The analysis of these zeta functions and of the indicated applications in physics is heavily based on this explicit knowledge in that well-known summation formulas are used.

In most cases, however, an explicit knowledge of the eigenvalues of, say, a Laplacian will not be available and an analysis of the associated zeta
functions will be more complicated. In recent years a new class of examples where eigenvalues are defined implicitly as solutions to transcendental equations has become accessible. In some detail let us assume that eigenvalues are determined by equations of the form
\beq F_\ell ( \lambda_{\ell ,n}) =0\label{implieigen}\eeq with $\ell ,n$ suitable indices. For example when trying to find eigenvalues and eigenfunctions of the Laplacian whenever possible one resorts to separation of variables and $\ell$ and $n$ would be suitable 'quantum numbers' labeling eigenfunctions. To be specific consider a scalar field in a three dimensional ball of radius $R$ with
Dirichlet boundary conditions. The eigenvalues $\lambda_k$ for
this situation, with $k$ as a multiindex, are thus determined
through
$$-\Delta \phi_k (x) = \lambda_k \phi _k (x), \quad \quad \phi (x)
|_{|x| = R} =0 .$$ In terms of spherical coordinates $(r,\Omega)$,
a complete set of eigenfunctions may be given in the form
$$\phi_{l,m,n} (r, \Omega ) = r^{-1/2} J_{l+1/2}
(\sqrt{\lambda_{l,n} }r ) Y _{l,m} (\Omega) ,$$ where $Y_{l,m}
(\Omega )$ are spherical surface harmonics \cite{erde55b}, and
$J_\nu$ are Bessel functions of the first kind \cite{grad65b}.
Eigenvalues of the Laplacian are determined as zeroes of Bessel
functions. In particular, for a given angular momentum quantum
number $l$, imposing Dirichlet boundary conditions, eigenvalues
$\lambda_{l,n}$ are determined by \beq J_{l+1/2} \left(
\sqrt{\lambda_{l,n} }R\right) =0 .\label{1}\eeq
Although some properties of the zeroes of Bessel functions are well understood \cite{grad65b},
there is no closed form for them available and
we encounter the situation described by (\ref{implieigen}). In order to find properties of the zeta function associated with this kind of boundary value problems the idea is to use the argument principle or Cauchy's residue theorem. For the situation of the ball one writes the zeta function in the form \beq \zeta (s) = \sum_{l=0}^\infty (2l+1) \frac 1 {2\pi i}
\int\limits_\gamma  k^{-2s} \frac
\partial {\partial k} \ln J_{l+1/2} (kR) dk,\label{2}\eeq
where the contour $\gamma$ runs counterclockwise and must enclose
all solutions of (\ref{1}). The factor $(2l+1)$ represents the
degeneracy for each angular momentum $l$ and the summation is over
all angular momenta. The integrand has singularities exactly at the eigenvalues and one can show that the residues are one such that the definition of the zeta function is recovered. More generally, in other coordinate systems, one would have, somewhat symbolically,
\beq\zeta (s) = \sum_j d_j \frac 1 {2\pi i}
\int\limits_\gamma k^{-2s} \frac \partial
{\partial k} \ln F_j (k) dk, \eeq
the task being to
construct the analytical continuation of this object.
The details of the procedure will depend very much
on the properties of the special function $F_j$ that enters, but often all the information needed can be found \cite{kirs02b}.
Nevertheless, for many separable coordinate systems this program has not been performed but efforts are being made in order to obtain yet unknown precise values for the Casimir energy for various geometries.

\section*{Acknowledgements}
This work is supported by the National Science Foundation Grant
PHY-0757791.
Part of the work was done while the author
enjoyed the hospitality and partial support of the
Department of Physics and Astronomy of
the University of Oklahoma. Thanks go in particular to Kimball Milton and his group who
made this very pleasant and exciting visit possible.

\def\cprime{$'$} \def\polhk#1{\setbox0=\hbox{#1}{\ooalign{\hidewidth
  \lower1.5ex\hbox{`}\hidewidth\crcr\unhbox0}}}



\end{document}